\documentclass{llncs}
\usepackage{papsi}

\title{Probabilistic Process Algebra \\ and Strategic Interleaving}
\author{C.A. Middelburg}
\institute{Informatics Institute, Faculty of Science, University of
           Amsterdam, \\
           Science Park~904, 1098~XH Amsterdam, the Netherlands \\
           \email{C.A.Middelburg@uva.nl}}

\begin{document}
\maketitle

\begin{abstract}
%% 84 %%
We first present a probabilistic version of \ACP\ that rests on the 
principle that probabilistic choices are always resolved before choices 
involved in alternative composition and parallel composition are 
resolved and then extend this probabilistic version of \ACP\ with a form 
of interleaving in which parallel processes are interleaved according to 
what is known as a process-scheduling policy in the field of operating 
systems. 
We use the term strategic interleaving for this more constrained form of 
interleaving.
The extension covers probabilistic process-scheduling \mbox{policies}.
\begin{classcode}
D.1.3, D.4.1, F.1.2
\end{classcode}
\end{abstract}

\section{Introduction}
\label{sect-intro}

First of all, we present a probabilistic version of 
\ACP~\cite{BW90,BK84a}, called \pACP\ (probabilistic ACP).
\pACP\ is a minor variant of the subtheory of \pACPt~\cite{AG09a} in 
which the operators for abstraction from some set of actions are 
lacking.
It is a minor variant of that subtheory because we take functions 
whose range is the carrier of a signed cancellation meadow instead of a 
field as probability measures, add probabilistic choice operators for 
the probabilities $0$ and $1$, and have an additional axiom because of 
the inclusion of these operators.
The probabilistic choice operators for the probabilities $0$ and $1$ 
cause no problem because a meadow has a total multiplicative inverse 
operation where the multiplicative inverse of zero is zero.
Because of this property, we could also improve the operational 
semantics of \pACP.
In particular, we could reduce the number of rules for the operational 
semantics and replace all negative premises by positive premises in the 
remaining rules. 

We also extend \pACP\ with a form of interleaving in which parallel 
processes are interleaved according to what is known as a 
process-scheduling policy in the field of operating systems
(see e.g.~\cite{SGG18a,TB15a}). 
In~\cite{BM17b}, we have extended \ACP\ with this more constrained form 
of interleaving.
In that paper, we introduced the term strategic interleaving for this 
form of interleaving and the term interleaving strategy for 
process-scheduling policy.
Unlike in the extension presented in~\cite{BM17b}, probabilistic 
interleaving strategies are covered in the extension presented in the
current paper.
More precisely, the latter extension assumes a generic interleaving 
strategy that can be instantiated with different specific interleaving 
strategies, including probabilistic ones.

A main contribution of this paper to the area of probabilistic 
process algebra is a semantics of \pACP\ for which the axioms of 
\pACP\ are sound and complete. 
For \pACPt, such a semantic is not available.
For \pTCPt, a variant of \pACPt, an erroneous semantics is given 
in~\cite{Geo11a} (see Section~\ref{subsect-remarks} for details).
This rules out the possibility to derive a semantics of \pACP\ or 
\pACPt\ from this semantics of \pTCPt. 
Another contribution of this paper is an extension of \pACP\ with
strategic interleaving that covers probabilistic interleaving 
strategies.
The work presented in~\cite{BM17b} and this paper is the only work on 
strategic interleaving in the setting of a general algebraic theory of 
processes like ACP, CCS and CSP.

The motivation for elaborating upon the work on \pACPt\ presented 
in~\cite{AG09a} is that it introduces a parallel composition operator
characterized by remarkably simple and natural axioms --- axioms that 
should be backed up by an appropriate semantics. 
The motivation for considering strategic interleaving in the setting of
\ACP\ originates from an important feature of many contemporary 
programming languages, namely multi-threading 
(see Section~\ref{subsect-motivation} for details).

The rest of this paper is organized as follows.
First, the theory of signed cancellation meadows is briefly summarized  
(Section~\ref{sect-meadows}).
Next, \pACP\ and its extension with guarded recursion, called \pACPr, 
is presented (Section~\ref{sect-pACPr}).
After that, the extension of \pACPr\ with strategic interleaving is 
presented (Section~\ref{sect-pSI}). 
Finally, we make some concluding remarks (Section~\ref{sect-concl}).

\section{Signed Cancellation Meadows}
\label{sect-meadows}

Later in this paper, we will take functions whose range is the carrier 
of a signed cancellation meadow as probability measures.
Therefore, we briefly summarize the theory of signed cancellation 
meadows in this section.

In~\cite{BT07a}, meadows are proposed as alternatives for fields with a
purely equational axiomatization.
Meadows are commutative rings with a multiplicative identity element and 
a total multiplicative inverse operation where the multiplicative 
inverse of zero is zero.
Fields whose multiplicative inverse operation is made total by imposing 
that the multiplicative inverse of zero is zero are called 
zero-totalized fields.
All zero-totalized fields are meadows, but not conversely.

\sloppy
% !!
Cancellation meadows are meadows that satisfy the 
\emph{cancellation axiom} 
\mbox{$ x \neq 0 \Land x \mul y = x \mul z \Limpl y = z$}.
The cancellation meadows that satisfy in addition the 
\emph{separation axiom} $0 \neq 1$ are exactly the zero-totalized 
fields.

Signed cancellation meadows are cancellation meadows expanded with a
signum operation.
The signum operation makes it possible that the predicates $<$ and 
$\leq$ are defined (see below).

The signature of signed cancellation meadows consists of the following
constants and operators:
\begin{itemize}
\item
the \emph{additive identity} constant $0$;
\item
the \emph{multiplicative identity} constant $1$;
\item
the binary \emph{addition} operator ${} +$ {};
\item
the binary \emph{multiplication} operator ${} \mul {}$;
\item
the unary \emph{additive inverse} operator $- {}$\,;
\item
the unary \emph{multiplicative inverse} operator ${}\minv$\,;
\item
the unary \emph{signum} operator $\sign$.
\end{itemize}

Terms are build as usual.
We use prefix notation, infix notation, and postfix notation as usual.
We also use the usual precedence convention.
We introduce subtraction and division as abbreviations:
$t - t'$ abbreviates $t + (-t')$ and
$t / t'$ abbreviates $t \mul ({t'}\minv)$.

Signed cancellation meadows are axiomatized by the equations in 
Tables~\ref{eqns-meadow} and~\ref{eqns-signum} and the above-mentioned
cancellation axiom.
\begin{table}[!t]
\caption
{Axioms of a meadow}
\label{eqns-meadow}
\begin{eqntbl}
\begin{eqncol}
(x + y) + z = x + (y + z)                                             \\
x + y = y + x                                                         \\
x + 0 = x                                                             \\
x + (-x) = 0
\end{eqncol}
\qquad\quad
\begin{eqncol}
(x \mul y) \mul z = x \mul (y \mul z)                                 \\
x \mul y = y \mul x                                                   \\
x \mul 1 = x                                                          \\
x \mul (y + z) = x \mul y + x \mul z
\end{eqncol}
\qquad\quad
\begin{eqncol}
(x\minv)\minv = x                                                     \\
x \mul (x \mul x\minv) = x                           
\end{eqncol}
\end{eqntbl}
\end{table}

\begin{table}[!t]
\caption{Additional axioms for the signum operator}
\label{eqns-signum}
\begin{eqntbl}
\begin{eqncol}
\sign(x / x) = x / x                                                  \\
\sign(1 - x / x) = 1 - x / x                                          \\
\sign(-1) = -1
\end{eqncol}
\qquad\quad
\begin{eqncol}
\sign(x\minv) = \sign(x)                                              \\
\sign(x \mul y) = \sign(x) \mul \sign(y)                              \\
(1 - \frac{\sign(x) - \sign(y)}{\sign(x) - \sign(y)}) \mul
(\sign(x + y) - \sign(x)) = 0
\end{eqncol}
\end{eqntbl}
\end{table}

The predicates $<$ and $\leq$ are defined in signed cancellation meadows
as follows: 
\begin{ldispl}
x < y \Liff \sign(y - x) = 1\;,
\\
x \leq y \Liff \sign(\sign(y - x) + 1) = 1\;.
\end{ldispl}%
Because $\sign(\sign(y - x) + 1) \neq -1$, we have
$0 \leq x \leq 1 \Liff
 \sign(\sign(x) + 1) \mul \sign(\sign(1 - x) + 1) = 1$.
We will use this equivalence below to describe the set of probabilities.

In~\cite{BP13a}, Kolmogorov's probability axioms for finitely additive 
probability spaces are rephrased for the case where probability measures 
are functions whose range is the carrier of a signed cancellation 
meadow.

\section{\pACP\ with Guarded Recursion}
\label{sect-pACPr}

In this section, we introduce \pACP\ (probabilistic Algebra of 
Communicating Processes) and guarded recursion in the setting of \pACP.
The algebraic theory \pACP\ is a minor variant of the subtheory of 
\pACPt~\cite{AG09a} in which the operators for abstraction from some set 
of actions are lacking.
\pACP\ is a variant of that subtheory because:
(a)~the range of the functions that are taken as probability measures is 
the carrier of a signed cancellation meadow in \pACP\ and the carrier of 
a field in \pACPt;
(b)~probabilistic choice operators for the probabilities~$0$ and~$1$,  
together with an axiom concerning the these two operators, are found in 
\pACP, but not in \pACPt.
Moreover, a semantics is available for \pACP, but not really for 
\pACPt.%
\footnote
{Issues with the semantics of \pACPt\ are discussed in 
 Section~\ref{subsect-remarks}.}

\subsection{\pACP}
\label{subsect-pACP}

In \pACP, it is assumed that a fixed but arbitrary set $\Act$ of 
\emph{actions}, with $\dead \notin \Act$, has been given.
We write $\Actd$ for $\Act \union \set{\dead}$.
Related to this, it is assumed that a fixed but arbitrary commutative 
and associative \emph{communication} function 
$\funct{\commf}{\Actd \x \Actd}{\Actd}$, with 
$\commf(\dead,a) = \dead$ for all $a \in \Actd$, has been given.
The function $\commf$ is regarded to give the result of synchronously
performing any two actions for which this is possible, and to give 
$\dead$ otherwise.

It is also assumed that a fixed but arbitrary signed cancellation 
meadow $\fM$ has been given.
We denote the interpretations of the constants and operators of signed 
cancellation meadows in $\fM$ by the constants and operators themselves.
We write $\Prob$ for the set
$\set{\pi \in \fM \where
 \sign(\sign(\pi) + 1) \mul \sign(\sign(1 - \pi) + 1) = 1}$
of \emph{probabilities}.

The signature of \pACP\ consists of the following constants and 
operators:
\begin{itemize}
\item
for each $a \in \Act$, the \emph{action} constant 
$a$\,;
\item
the \emph{inaction} constant $\dead$\,;
\item
the binary \emph{alternative composition} operator 
${} \altc {}$\,;
\item
the binary \emph{sequential composition} operator 
${} \seqc {}$\,;
\item
for each $\pi \in \Prob$,
the binary \emph{probabilistic choice} operator 
${} \paltc{\pi} {}$\,;
\item
the binary \emph{parallel composition} operator 
${} \parc {}$\,;
\item
the binary \emph{left merge} operator 
${} \leftm {}$\,;
\item
the binary \emph{communication merge} operator 
${} \commm {}$\,;
\item
for each $H \subseteq \Act$, the unary \emph{encapsulation} operator
$\encap{H}$\,.
\end{itemize}
We assume that there is a countably infinite set $\cX$ of variables, 
which contains $x$, $y$ and $z$, with and without subscripts.
Terms are built as usual.
We use infix notation for the binary operators.
The precedence conventions used with respect to the operators of \pACP\
are as follows: $\altc$ binds weaker than all others, $\seqc$ binds
stronger than all others, and the remaining operators bind equally
strong.

The constants and operators of \pACP\ can be explained as follows:
\begin{itemize}
\item
the constant $a$ denotes the process that can only perform action $a$ 
and after that terminate successfully;
\item
the constant $\dead$ denotes the process that cannot do anything;
\item
a closed term of the form $t \altc t'$ denotes the process that can 
behave as the process denoted by $t$ or as the process denoted by $t'$, 
where the choice between the two is resolved exactly when the first 
action of one of them is performed;
\item
a closed term of the form $t \seqc t'$ denotes the process that can 
first behave as the process denoted by $t$ and can next behave as the 
process denoted by $t'$;
\item
a closed term of the form $t \paltc{\pi} t'$ denotes the process that 
will behave as the process denoted by $t$ with probability $\pi$ and as 
the process denoted by $t'$ with probability $1 - \pi$, where the choice 
between the two processes is resolved before the first action of one of 
them is performed; 
\item
a closed term of the form $t \parc t'$ denotes the process that can 
behave as the process that proceeds with the processes denoted by $t$ 
and $t'$ in parallel;
\item
a closed term of the form $t \leftm t'$ denotes the process that can 
behave the same as the process denoted by $t \parc t'$, except that it 
starts with performing an action of the process denoted by $t$;
\item
a closed term of the form $t \commm t'$ denotes the process that can 
behave the same as the process denoted by $t \parc t'$, except that it 
starts with performing an action of the process denoted by $t$ and an 
action of the process denoted by~$t'$ synchronously;
\item
a closed term of the form $\encap{H}(t)$ denotes the process that can
behave the same as the process denoted by $t$, except that actions from 
$H$ are blocked.
\end{itemize}
Processes in parallel are considered to be arbitrarily interleaved. 
With that, probabilistic choices are resolved before interleaving steps 
are enacted.

The operators $\leftm$ and $\commm$ are of an auxiliary nature.
They are needed to axiomatize \pACP.

The axioms of \pACP\ are the equations given in Table~\ref{axioms-pACP}.
\begin{table}[!t]
\caption{Axioms of \pACP}
\label{axioms-pACP}
\begin{eqntbl}
\begin{axcol}
x \altc y = y \altc x                                  & \axiom{A1}   \\
(x \altc y) \altc z = x \altc (y \altc z)              & \axiom{A2}   \\
a \altc a = a                                          & \axiom{A3'}  \\
(x \altc y) \seqc z = x \seqc z \altc y \seqc z        & \axiom{A4}   \\
(x \seqc y) \seqc z = x \seqc (y \seqc z)              & \axiom{A5}   \\
x \altc \dead = x                                      & \axiom{A6}   \\
\dead \seqc x = \dead                                  & \axiom{A7}   \\
{}                                                                    \\
{}                                                                    \\
\encap{H}(a) = a                \hfill \mif a \notin H & \axiom{D1}   \\
\encap{H}(a) = \dead            \hfill \mif a \in H    & \axiom{D2}   \\
\encap{H}(x \altc y) = \encap{H}(x) \altc \encap{H}(y) & \axiom{D3}   \\
\encap{H}(x \seqc y) = \encap{H}(x) \seqc \encap{H}(y) & \axiom{D4}   \\
\multicolumn{2}{@{}c@{}}{\dotfill} \\[1ex] 
x \paltc{\pi} y = y \paltc{1{-}\pi} x                  & \axiom{pA1}  \\
(x \paltc{\pi} y) \paltc{\rho} z = \\ \quad\;\;
x \paltc{\pi{\mul}\rho}
(y \paltc{\frac{(1{-}\pi){\mul}\rho}{1{-}\pi{\mul}\rho}} z)
                                                       & \axiom{pA2}  \\
x \paltc{\pi} x = x                                    & \axiom{pA3}  \\
(x \paltc{\pi} y) \seqc z = x \seqc z \paltc{\pi} y \seqc z
                                                       & \axiom{pA4}  \\
(x \paltc{\pi} y) \altc z = (x \altc z) \paltc{\pi} (y \altc z)
                                                       & \axiom{pA5}  \\
x \paltc{1} y = x                                      & \axiom{pA6}  
\end{axcol}
\quad\;\; 
\begin{axcol}
x = x \altc x \Land y = y \altc y \Limpl \\ \quad\;\;
x \parc y =
          x \leftm y \altc y \leftm x \altc x \commm y & \axiom{CM1'}\\
a \leftm x = a \seqc x                                 & \axiom{CM2}  \\
a \seqc x \leftm y = a \seqc (x \parc y)               & \axiom{CM3}  \\
(x \altc y) \leftm z = x \leftm z \altc y \leftm z     & \axiom{CM4}  \\
a \seqc x \commm b = \commf(a,b) \seqc x               & \axiom{CM5}  \\
a \commm b \seqc x = \commf(a,b) \seqc x               & \axiom{CM6}  \\
a \seqc x \commm b \seqc y =
                        \commf(a,b) \seqc (x \parc y)  & \axiom{CM7}  \\
(x \altc y) \commm z = x \commm z \altc y \commm z     & \axiom{CM8}  \\
x \commm (y \altc z) = x \commm y \altc x \commm z     & \axiom{CM9}  \\
\dead \commm x = \dead                                 & \axiom{CM10} \\
x \commm \dead = \dead                                 & \axiom{CM11} \\
a \commm b = \commf(a,b)                               & \axiom{CM12} \\
\multicolumn{2}{@{}c@{}}{\dotfill} \\[1ex] 
(x \paltc{\pi} y) \parc z = (x \parc z) \paltc{\pi} (y \parc z)
                                                       & \axiom{pCM1} \\                                                               
x \parc (y \paltc{\pi} z) = (x \parc y) \paltc{\pi} (x \parc z)
                                                       & \axiom{pCM2} \\
(x \paltc{\pi} y) \leftm z = (x \leftm z) \paltc{\pi} (y \leftm z)
                                                       & \axiom{pCM3} \\                                                               
x \leftm (y \paltc{\pi} z) = (x \leftm y) \paltc{\pi} (x \leftm z)
                                                       & \axiom{pCM4} \\
(x \paltc{\pi} y) \commm z = (x \commm z) \paltc{\pi} (y \commm z)
                                                       & \axiom{pCM5} \\                                                               
x \commm (y \paltc{\pi} z) = (x \commm y) \paltc{\pi} (x \commm z)
                                                       & \axiom{pCM6} \\
\encap{H}(x \paltc{\pi} y) = \encap{H}(x) \paltc{\pi} \encap{H}(y)
                                                       & \axiom{pD}  
\end{axcol}
\end{eqntbl}
\end{table}
In these equations, $a$ and $b$ stand for arbitrary constants of \pACP\
(which include the action constants and the inaction constant), 
$H$ stands for an arbitrary subset of $\Act$, and $\pi$ and $\rho$ stand
for arbitrary probabilities from $\Prob$.
Moreover, $\commf(a,b)$ stands for the action constant for the action 
$\commf(a,b)$.
In D1 and D2, side conditions restrict what $a$ and $H$ stand for.

The equations in Table~\ref{axioms-pACP} above the dotted lines, with
$\axiom{A3'}$ replaced by the equation $x \altc x = x$ and 
$\axiom{CM1'}$ replaced by its consequent, constitute an axiomatization 
of \ACP.
In presentations of \ACP, $\commf(a,b)$ is regularly replaced by 
$a \commm b$ in CM5--CM7.
By CM12, which is more often called CF, these replacements give rise to
an equivalent axiomatization.
Moreover, CM10 and CM11 are usually absent.
These equations are not derivable from the other axioms, but all their 
closed substitution instances are derivable from the other axioms and 
they hold in all models that have been considered for \ACP\ in the 
literature.

With regard to axiom $\axiom{CM1'}$, we remark that, 
for each closed term $t$ of \pACP\ that is not derivably equal to a term
of the form $t' \paltc{\pi} t''$ with $\pi \in \Prob \diff \set{0,1}$, 
$t = t \altc t$ is derivable. 
In other words, if the process denoted by $t$ is not initially 
probabilistic in nature, then $t = t \altc t$ is derivable.

\pACP\ has pA1, pA3--pA5, pCM1--pCM2, and pD in common with \pACPt\ as 
presented in~\cite{AG09a}.
Replacement of axiom pA2 of \pACP\ by axiom pA2 of \pACPt, that is
$x \paltc{\pi} (y \paltc{\rho} z) = 
 (x \paltc{\frac{\pi}{\pi{+}\rho{-}\pi{\mul}\rho}} y)
  \paltc{\pi{+}\rho{-}\pi{\mul}\rho} z$,
gives rise to an equivalent axiomatization.
In~\cite{Geo11a}, axioms pCM3--pCM6 are presented as axioms of \pTCPt, 
a variant of \pACPt\ in which the action constants have been replaced by 
action prefixing operators and a constant for the process that is only 
capable of terminating successfully.
Therefore, axioms pCM3--pCM6 may be absent in~\cite{AG09a} by mistake.

Axiom pA6 is new.
Notice that $(x \paltc{0} y) \paltc{0} z = z$ and 
$x \paltc{0} (y \paltc{0} z) = z$ are derivable from pA1 and pA6. 
This is consistent with the instance of pA2 where $\pi = \rho = 0$
because in meadows $0 / 0 = 0$.

In the sequel, we will use the notation $\Altc{i=1}{n} t_i$, where 
$n \geq 1$, for right-nested alternative compositions.
For each $n \in \Natpos$, 
the term $\Altc{i = 1}{n} t_i$ is defined by induction on $n$ as 
follows:%
\footnote
{We write $\Natpos$ for the set $\set{n \in \Nat \where n \geq 1}$ of 
 positive natural numbers.}
\begin{ldispl}
\Altc{i = 1}{1} t_i = t_1 
\quad \mathrm{and} \quad
\Altc{i = 1}{n + 1} t_i = t_1 \altc \Altc{i = 1}{n} t_{i+1}\;.
\end{ldispl}%
In addition, we will use the convention that 
$\Altc{i = 1}{0} t_i = \dead$.

In the sequel, we will also use the notation 
$\PRC{i = 1}{n}{\pi_i}{t_i}$ where $n \geq 1$ and  
$\sum_{i < n} \pi_i = 1$, for right-nested probabilistic choices.
For each $n \in \Natpos$, the term $\PRC{i = 1}{n}{\pi_i}{t_i}$ is 
defined by induction on $n$ as follows:
\begin{ldispl}
\PRC{i = 1}{1}{\pi_i}{t_i}   = t_1
\quad \mathrm{and} \quad
\PRC{i = 1}{n+1}{\pi_i}{t_i} =
t_1 \paltc{\pi_1} 
(\PRC{i = 1}{n}{\frac{\pi_{i+1}}{1 - \pi_1}}{t_{i+1}})\;.
\end{ldispl}%
The process denoted by $\PRC{i = 1}{n + 1}{\pi_i}{t_i}$ will behave like 
the process denoted by $t_1$ with probability $\pi_1$, \ldots,  
like the process denoted by $t_{n+1}$ with probability $\pi_{n+1}$.

In the next definition, the following summand notation is used.
Let $t$ and $t'$ be closed \pACP\ terms.
Then we write $t \summ t'$ for the assertion that $t \equiv t'$ or there 
exists a closed \pACP\ term $t''$ such that $t \altc t'' = t'$ is 
derivable from axioms A1 and A2 
and we write $t \psumm t'$ for the assertion that $t \equiv t'$ or there 
exists a closed \pACP\ term $t''$ and a $\pi\in \Prob \diff \set{0,1}$ 
such that $t \paltc{\pi} t'' = t'$ is derivable from axioms pA1 and 
pA2.%
\footnote
{We write $t \equiv t'$ to indicate that $t$ is syntactically equal to
 $t'$.}

Each closed \pACP\ term is derivably equal to a proper basic term of 
\pACP. \linebreak[2]
The set $\cB$ of \emph{proper basic terms of} \pACP\ is inductively 
defined, simultaneously with auxiliary sets $\cB^0$, $\cB^1$, $\cB^2$, 
and $\cB^3$, by the following rules:
\begin{itemize}
\item
$\dead \in \cB^0$;
\item
if $a \in \Act$, then $a \in \cB^1$;
\item
if $a \in \Act$ and $t \in \cB$, then $a \seqc t \in \cB^1$;
\item
if $t \in \cB^1$, then $t \in \cB^2$;
\item
if $t \in \cB^1$, $t' \in \cB^2$, and not $t \summ t'$,
then $t \altc t' \in \cB^2$;
\item
if $t \in \cB^2$, then $t \in \cB^3$;
\item
if $t \in \cB^2$, $t' \in \cB^3$, not $t \psumm t'$, and 
$\pi \in \Prob \diff \set{0,1}$, then $t \paltc{\pi} t' \in \cB^3$;
\item
if $t \in \cB^0$, then $t \in \cB$;
\item
if $t \in \cB^3$, then $t \in \cB$.
\end{itemize}
\begin{proposition}
\label{prop-elim-basic}
For each \pACP\ term $t$, there exists a proper basic term $t'$ of \pACP\ 
such that $t = t'$ is derivable from the axioms of \pACP.
\end{proposition}
\begin{proof}
The proof is straightforward by induction on the structure of $t$.
The case where $t$ is of the form $\dead$ and the case where $t$ is of 
the form $a$ ($a \in \Act$) are trivial.
The case where $t$ is of the form $t_1 \seqc t_2$ follows immediately 
from the induction hypothesis (applied to $t_1$ and $t_2$) and the claim 
that, for all proper basic terms $t_1'$ and $t_2'$ of \pACP, there 
exists a proper basic term $t'$ of \pACP\ such that 
$t_1' \seqc t_2' = t'$ is derivable from the axioms of \pACP.
This claim is straightforwardly proved by induction on the structure of 
$t_1'$.
The cases where $t$ is of the form $t_1 \altc t_2$, 
$t_1 \paltc{\pi} t_2$, $t_1 \leftm t_2$, $t_1 \commm t_2$ or 
$\encap{H}(t_1)$ are proved in the same vein as the case where $t$ is of 
the form $t_1 \seqc t_2$.
In the case that $t$ is of the form $t_1 \commm t_2$, each of the cases 
to be considered in the inductive proof of the claim demands a (nested) 
proof by induction on the structure of $t_2'$.
The case that $t$ is of the form $t_1 \parc t_2$ follows immediately 
from the case that $t$ is of the form $t_1 \leftm t_2$ and the case that 
$t$ is of the form $t_1 \commm t_2$.
\qed
\end{proof}

\subsection{Guarded Recursion}
\label{subsect-recursion}

A closed \pACP\ term denotes a process with a finite upper bound to the 
number of actions that it can perform. 
Guarded recursion allows the description of processes without a finite 
upper bound to the number of actions that it can perform.

The current subsection applies to both \pACP\ and its extension \pACPpSI\ 
introduced in Section~\ref{sect-pSI}.  
Therefore, in the current subsection, let \PPA\ be \pACP\ or \pACPpSI. 

Let $t$ be a \PPA\ term containing a variable $X$.
Then an occurrence of $X$ in $t$ is \emph{guarded} if $t$ has a subterm 
of the form $a \seqc t'$ where $a \in \Act$ and $t'$ is a \PPA\ term 
containing this occurrence of $X$.
A \PPA\ term $t$ is a \emph{guarded} \PPA\ term if all occurrences of 
variables in $t$ are guarded.

A \emph{recursive specification} over \PPA\ is a set 
$\set{X_i = t_i \where i \in I}$, where $I$ is a finite or countably 
infinite set, each $X_i$ is a variable from $\cX$, each $t_i$ is a \PPA\ 
term in which only variables from $\set{X_i \where i \in I}$ occur, and 
$X_i \neq X_j$ for all $i,j \in I$ with $i \neq j$.
A recursive specification $\set{X_i = t_i \where i \in I}$ over \PPA\ is 
a \emph{guarded} recursive specification over \PPA\ if each $t_i$ is 
rewritable to a guarded \PPA\ term  using the axioms of \PPA\ in either 
direction and the equations in 
$\set{X_j = t_j \where j \in I \Land i \neq j}$ from left to right.

We write $\vars(E)$, where $E$ is a guarded recursive specification, for 
the set of all variables that occur in $E$.
The equations occurring in a guarded recursive specification are called
\emph{recursion equations}.

A solution of a guarded recursive specification $E$ in some model of 
\PPA\ is a set $\set{P_X \where X \in \vars(E)}$ of elements of the 
carrier of that model such that the equations of $E$ hold if, for all 
$X \in \vars(E)$, $X$ is assigned $P_X$.
We are only interested in models of \PPA\ in which guarded recursive 
specifications have unique solutions --- such as the model presented in
Section~\ref{subsect-semantics}.

We extend \PPA\ with guarded recursion by adding constants for solutions 
of guarded recursive specifications over \PPA\ and axioms concerning 
these additional constants.
For each guarded recursive specification $E$ over \PPA\ and each 
$X \in \vars(E)$, we add a constant standing for the unique solution of 
$E$ for $X$ to the constants of \PPA.
The constant standing for the unique solution of $E$ for $X$ is denoted 
by $\rec{X}{E}$.
We use the following notation.
Let $t$ be a \PPA\ term and $E$ be a guarded recursive specification 
over \PPA.
Then we write $\srec{t}{E}$ for $t$ with, for all $X \in \vars(E)$, all
occurrences of $X$ in $t$ replaced by $\rec{X}{E}$.
We add the equation RDP and the conditional equation RSP given in 
Table~\ref{axioms-recursion} to the axioms of \PPA.
\begin{table}[!t]
\caption{Axioms for guarded recursion}
\label{axioms-recursion}
\begin{eqntbl}
\begin{saxcol}
\rec{X}{E} = \srec{t}{E} & \mif X = t\; \in \;E         & \axiom{RDP} \\
E \Limpl X = \rec{X}{E}  & \mif X \in \vars(E)          & \axiom{RSP} 
\end{saxcol}
\end{eqntbl}
\end{table}
In RDP and RSP, $X$ stands for an arbitrary variable from $\cX$, $t$ 
stands for an arbitrary \PPA\ term, and $E$ stands for an arbitrary 
guarded recursive specification over \PPA.
Side conditions restrict what $X$, $t$ and $E$ stand for.
We write $\PPA_\mathrm{rec}$ for the resulting theory.

The equations $\rec{X}{E} = \srec{t}{E}$ for a fixed $E$ express that 
the constants $\rec{X}{E}$ make up a solution of $E$. 
The conditional equations $E \Limpl X = \rec{X}{E}$ express that this 
solution is the only one.

Because we have to deal with conditional equational formulas with an 
countably infinite number of premises in $\PPA_\mathrm{rec}$, it is 
understood that infinitary conditional equational logic is used in 
deriving equations from the axioms of $\PPA_\mathrm{rec}$. 
A complete inference system for infinitary conditional equational logic 
can be found in, for example, \cite{GV93}. 
It is noteworthy that in the case of infinitary conditional equational 
logic derivation trees may be infinitely branching (but they may not 
have infinite branches).

\subsection{Semantics of \pACP\ with Guarded Recursion}
\label{subsect-semantics}

In this subsection, we present a structural operational semantics of 
\pACPr\ and define a notion of bisimulation equivalence based on this 
semantics.

We start with the presentation of a structural operational semantics 
of \pACPr.
The following relations on closed \pACPr\ terms are used:
\begin{itemize}
\item 
for each $a \in \Act$, a unary relation $\aterm{}{a}$\,;
\item 
for each $a \in \Act$, a binary relation $\astep{}{a}{}$\,;
\item 
for each $\pi \in \Prob$, a binary relation $\pstep{}{\pi}{}$\,.
\end{itemize}
We write $\aterm{t}{a}$ for the assertion that $t \in {\aterm{}{a}}$, 
$\astep{t}{a}{t'}$ for the assertion that $(t,t') \in {\astep{}{a}{}}$, 
$\pstep{t}{\pi}{t'}$ for the assertion that 
$(t,t') \in {\pstep{}{\pi}{}}$, and
$\npstep{t}{t'}$ for the assertion that, for all 
$\pi \in \Prob \setminus \set{0}$, not $(t,t') \in {\pstep{}{\pi}{}}$.
These assertions can be explained as follows:
\begin{itemize}
\item
$\aterm{t}{a}$ indicates that 
$t$ can perform action $a$ and then terminate successfully;
\item
$\astep{t}{a}{t'}$ indicates that 
$t$ can perform action $a$ and then behave as $t'$;
\item
$\pstep{t}{\pi}{t'}$ indicates that
$t$ will behave as $t'$ with probability $\pi$;
\item
\smash{$\npstep{t}{t'}$} indicates that $t$ will not behave as $t'$ with 
a probability greater than zero.
\end{itemize}
The structural operational semantics of \pACPr\ is described by the
rules given in Tables~\ref{rules-pACPr-1} and~\ref{rules-pACPr-2}.
\begin{table}[!t]
\caption{Rules for the operational semantics of \pACPr\ (part 1)}
\label{rules-pACPr-1}
\begin{ruletbl}
\Rule
{}
{\aterm{a}{a}}
\qquad
\\
\Rule
{\aterm{x}{a},\; \pstep{y}{1}{y'}}
{\aterm{x \altc y}{a}}
\qquad
\Rule
{\pstep{x}{1}{x'},\; \aterm{y}{a}}
{\aterm{x \altc y}{a}}
\qquad
\Rule
{\astep{x}{a}{x'},\; \pstep{y}{1}{y'}}
{\astep{x \altc y}{a}{x'}}
\qquad
\Rule
{\pstep{x}{1}{x'},\; \astep{y}{a}{y'}}
{\astep{x \altc y}{a}{y'}}
\\
\Rule
{\aterm{x}{a}}
{\astep{x \seqc y}{a}{y}}
\qquad
\Rule
{\astep{x}{a}{x'}}
{\astep{x \seqc y}{a}{x' \seqc y}}
\\
\Rule
{\aterm{x}{a},\; \pstep{y}{1}{y'}}
{\astep{x \parc y}{a}{y}}
\qquad
\Rule
{\pstep{x}{1}{x'},\; \aterm{y}{a}}
{\astep{x \parc y}{a}{x}}
\qquad
\Rule
{\astep{x}{a}{x'},\; \pstep{y}{1}{y'}}
{\astep{x \parc y}{a}{x' \parc y}}
\qquad
\Rule
{\pstep{x}{1}{x'},\; \astep{y}{a}{y'}}
{\astep{x \parc y}{a}{x \parc y'}}
\\
\RuleC
{\aterm{x}{a},\; \aterm{y}{b}}
{\aterm{x \parc y}{\commf(a,b)}}
{\commf(a,b) \in \Act}
\qquad
\RuleC
{\aterm{x}{a},\; \astep{y}{b}{y'}}
{\astep{x \parc y}{\commf(a,b)}{y'}}
{\commf(a,b) \in \Act}
\\
\RuleC
{\astep{x}{a}{x'},\; \aterm{y}{b}}
{\astep{x \parc y}{\commf(a,b)}{x'}}
{\commf(a,b) \in \Act}
\qquad
\RuleC
{\astep{x}{a}{x'},\; \astep{y}{b}{y'}}
{\astep{x \parc y}{\commf(a,b)}{x' \parc y'}}
{\commf(a,b) \in \Act}
\\
\Rule
{\aterm{x}{a}}
{\astep{x \leftm y}{a}{y}}
\qquad
\Rule
{\astep{x}{a}{x'}}
{\astep{x \leftm y}{a}{x' \parc y}}
\\
\RuleC
{\aterm{x}{a},\; \aterm{y}{b}}
{\aterm{x \commm y}{\commf(a,b)}}
{\commf(a,b) \in \Act}
\qquad
\RuleC
{\aterm{x}{a},\; \astep{y}{b}{y'}}
{\astep{x \commm y}{\commf(a,b)}{y'}}
{\commf(a,b) \in \Act}
\\
\RuleC
{\astep{x}{a}{x'},\; \aterm{y}{b}}
{\astep{x \commm y}{\commf(a,b)}{x'}}
{\commf(a,b) \in \Act}
\qquad
\RuleC
{\astep{x}{a}{x'},\; \astep{y}{b}{y'}}
{\astep{x \commm y}{\commf(a,b)}{x' \parc y'}}
{\commf(a,b) \in \Act}
\\
\RuleC
{\aterm{x}{a}}
{\aterm{\encap{H}(x)}{a}}
{a \not\in H}
\qquad
\RuleC
{\astep{x}{a}{x'}}
{\astep{\encap{H}(x)}{a}{\encap{H}(x')}}
{a \not\in H}
\\
\RuleC
{\aterm{\srec{t}{E}}{a}}
{\aterm{\rec{X}{E}}{a}}
{X = t\; \in \;E}
\qquad
\RuleC
{\astep{\srec{t}{E}}{a}{x'}}
{\astep{\rec{X}{E}}{a}{x'}}
{X = t\; \in \;E}
\end{ruletbl}
\end{table}
\begin{table}[!t]
\caption{Rules for the operational semantics of \pACPr\ (part 2)}
\label{rules-pACPr-2}
\begin{ruletbl}
\Rule
{}
{\pstep{a}{1}{a}}
\qquad
\Rule
{}
{\pstep{\dead}{1}{\dead}}
\\
\Rule
{\pstep{x}{\pi}{x'},\; \pstep{y}{\rho}{y'}}
{\pstep{x \altc y}{\pi \mul \rho}{x' \altc y'}}
\qquad
\Rule
{\pstep{x}{\pi}{x'}}
{\pstep{x \seqc y}{\pi}{x' \seqc y}}
\qquad
\Rule
{\pstep{x}{\rho}{z},\; \pstep{y}{\rho'}{z}}
{\pstep{x \paltc{\pi} y}{\pi \mul \rho + (1-\pi) \mul \rho'}{z}}
\\
\Rule
{\pstep{x}{\pi}{x'},\; \pstep{y}{\rho}{y'}}
{\pstep{x \parc y}{\pi \mul \rho}{x' \parc y'}}
\qquad
\Rule
{\pstep{x}{\pi}{x'},\; \pstep{y}{\rho}{y'}}
{\pstep{x \leftm y}{\pi \mul \rho}{x' \leftm y'}}
\qquad
\Rule
{\pstep{x}{\pi}{x'},\; \pstep{y}{\rho}{y'}}
{\pstep{x \commm y}{\pi \mul \rho}{x' \commm y'}}
\\
\Rule
{\pstep{x}{\pi}{x'}}
{\pstep{\encap{H}(x)}{\pi}{\encap{H}(x')}}
\qquad
\RuleC
{\pstep{\srec{t}{E}}{\pi}{z}}
{\pstep{\rec{X}{E}}{\pi}{z}}
{X = t\; \in \;E}
\\
\Rule
{\npstep{x}{x'}}
{\pstep{x}{0}{x'}}
\end{ruletbl}
\end{table}
The rules in Table~\ref{rules-pACPr-1} describe the relations 
$\aterm{}{a}$ and the relations $\astep{}{a}{}$ and the rules in 
Table~\ref{rules-pACPr-2} describe the relations $\pstep{}{\pi}{}$.
In these tables, $a$ and $b$ stand for arbitrary actions from $\Act$, 
$\pi$, $\rho$, and $\rho'$ stand for arbitrary probabilities from 
$\Prob$,
$X$ stands for an arbitrary variable from $\cX$, 
$t$ stands for an arbitrary \pACP\ term, and 
$E$ stands for an arbitrary guarded recursive specification over \pACP.

We could have excluded the relation $\pstep{}{0}{}$ and by that obviated
the need for the last rule in Table~\ref{rules-pACPr-2}.
In that case, however, 11 additional rules concerning the relations 
$\pstep{}{\pi}{}$, all with negative premises, would be needed instead.

Notice that, if $t$ is not derivably equal to a term whose outermost 
operator is a probabilistic choice operator, then $t$ can only behave as 
itself and consequently we have that $\pstep{t}{1}{t}$ and 
$\pstep{t}{0}{t'}$ for each term $t'$ other than $t$.
% In Table~\ref{rules-pACPr-1}, $\pstep{t}{1}{t}$

The next two propositions express properties of the relations $\pstep{}{\pi}{}$.
\begin{proposition}
\label{prop-pstep-1}
For all closed \pACPr\ terms $t$ and $t'$, $\pstep{t}{1}{t'}$ only if 
$t \equiv t'$.
\end{proposition}
\begin{proof}
This is easy to prove by induction on the structure of $t$.
\qed
\end{proof}
\begin{proposition}
\label{prop-pstep-2}
For all closed \pACPr\ terms $t$ and $t'$, there exists a 
$\pi \in \Prob$ such that $\pstep{t}{\pi}{t'}$.
\end{proposition}
\begin{proof}
This is easy to prove by induction on the structure of $t$.
\qed
\end{proof}

We define a probability distribution function $P$ from the set of all 
pairs of closed \pACPr\ terms to $\Prob$ as follows:
\pagebreak[2]
\begin{ldispl}
P(t,t') = {\displaystyle \sum_{\pi \in \Pi(t,t')}} \!\pi \;,
\quad\mathrm{where}\;\; 
\Pi(t,t') = \set{\pi \where \pstep{t}{\pi}{t'}}\;.
\end{ldispl}%
This function can be explained as follows:
$P(t,t')$ is the total probability that $t$ will behave as $t'$.

We write $P(t,T)$, where $t$ is a closed \pACPr\ term and $T$ is a set
of closed \pACPr\ terms, for $\sum_{t' \in T} P(t,t')$.

The well-definedness of $P$ is a corollary of 
Proposition~\ref{prop-pstep-2}.
\begin{corollary}
\label{corol-well-def-pdf}
For all closed \pACPr\ terms $t$ and $t'$, there exists a unique
$\pi \in \Prob$ such that $P(t,t') = \pi$.
\end{corollary}
Moreover, $P$ is actually a probability distribution function.
\begin{proposition}
\label{prop-pdf}
Let $T$ be the set of all closed \pACPr\ terms.
Then, for all closed \pACPr\ terms $t$, $P(t,T) = 1$.
\end{proposition}
\begin{proof}
This is easy to prove by induction on the structure of $t$.
\qed
\end{proof}

It follows from Propositions~\ref{prop-pstep-1} 
and~\ref{prop-pdf} that the behaviour of $t$ does not start 
with a probabilistic choice if $\pstep{t}{1}{t'}$.
This explains the premises $\pstep{x}{1}{x'}$ and $\pstep{y}{1}{y'}$ in
Table~\ref{rules-pACPr-1}: they guarantee that probabilistic choices are 
always resolved before choices involved in alternative composition and 
parallel composition are resolved.

The relations used in an operational semantics are often called 
transition relations.
It is questionable whether the relations $\pstep{}{\pi}{}$ deserve this 
name.
Recall that $\pstep{t}{\pi}{t'}$ means that $t$ will behave as $t'$ with 
probability $\pi$.
It is rather far-fetched to suppose that a transition from $t$ to $t'$ 
has taken place at the time that $t$ starts to behave as $t'$.
The relations $\pstep{}{\pi}{}$ primarily constitute a representation of 
the probability distribution function $P$ defined above.
This representation turns out to be a convenient one in the setting of 
structural operational semantics.

In the next paragraph, we write $[t]_R$, where $t$ is a closed \pACPr\ 
term and $R$ is an equivalence relation on closed \pACPr\ terms, for the 
equivalence class of $t$ with respect to $R$.

A \emph{probabilistic bisimulation} is an equivalence relation $R$ on 
closed \pACPr\ terms such that, for all closed \pACPr\ terms $t_1,t_2$ 
with $R(t_1,t_2)$, the following conditions hold: 
\begin{itemize}
\item
if $\astep{t_1}{a}{t_1'}$ for some closed \pACPr\ term $t_1'$ and 
$a \in \Act$, then there exists a closed \pACPr\ term $t_2'$ such that 
$\astep{t_2}{a}{t_2'}$ and $R(t_1',t_2')$;
\item
if $\aterm{t_1}{a}$ for some $a \in \Act$, then $\aterm{t_2}{a}$;
\item
$P(t_1,[t]_R) = P(t_2,[t]_R)$ for all closed \pACPr\ terms $t$.
\end{itemize}%
Two closed \pACPr\ terms $t_1,t_2$ are 
\emph{probabilistic bisimulation equivalent}, written $t_1 \bisim t_2$, 
if there exists a probabilistic bisimulation $R$ such that $R(t_1,t_2)$.
Let $R$ be a probabilistic bisimulation such that $R(t_1,t_2)$.
Then we say that $R$ is a probabilistic bisimulation \emph{witnessing}
$t_1 \bisim t_2$.

The next two propositions state some useful results about 
${} \bisim {}$.
\begin{proposition}
\label{prop-bisim-pstep-one}
For all closed \pACPr\ terms $t$, $t \bisim t \altc t$ only if 
$\pstep{t}{1}{t}$.
\end{proposition}
\begin{proof}
This follows immediately from the rules for the operational semantics of 
\pACPr, using that, for all $\pi \in \Prob$, $\pi \mul \pi = 1$ iff 
$\pi = 1$.
\qed
\end{proof}
\begin{proposition}
\label{prop-bisim-max}
${} \bisim {}$ is the maximal probabilistic bisimulation.
\end{proposition}
\begin{proof}
It follows from the definition of ${} \bisim {}$ that it is sufficient 
to prove that ${} \bisim {}$ is a probabilistic bisimulation.

We start with proving that ${} \bisim {}$ is an equivalence relation.
The proofs of reflexivity and symmetry are trivial.
Proving transitivity amounts to showing\linebreak[2] that the conditions 
from the definition of a probabilistic bisimulation hold for the 
composition of two probabilistic bisimulations.
The proofs that the conditions concerning the relations $\astep{}{a}{}$ 
and $\aterm{}{a}$ hold are trivial.
The proof that the condition concerning the function $P$ holds is also 
easy using the following easy-to-check property of $P$: 
if $I$ is an index set and, 
for each $i \in I$, $T_i$ is a set of closed \pACPr\ terms such that, 
for all $i,j \in I$ with $i \neq j$, $T_i \inter T_j = \emptyset$, then 
$P(t,\Union_{i \in I} T_i) = \sum_{i \in I} P(t,T_i)$.

We also have to prove that the conditions from the definition of a
probabilistic bisimulation hold for ${} \bisim {}$.
The proofs that the conditions concerning the relations $\astep{}{a}{}$ 
and $\aterm{}{a}$ hold are trivial.
The proof that the condition concerning the function $P$ holds is easy 
knowing the above-mentioned property of $P$.
\qed
\end{proof}

\subsection{Soundness and Completeness Results}
\label{subsect-sound-complete}

In this subsection, we present a soundness theorem for \pACPr\ and a 
completeness theorem for \pACP.

We write $R\eqvcl$, where $R$ is a binary relation, for the
equivalence closure of $R$.

The following proposition will be used below in the proof of a soundness 
theorem for \pACPr.
\begin{proposition}
\label{prop-bisim-congr}
${} \bisim {}$ is a congruence with respect to the operators of \pACPr.
\end{proposition}
\begin{proof}
In this proof, we write $R_1 \diamond R_2$, where $R_1$ and $R_2$ are 
probabilistic bisimulations and $\diamond$ is a binary operator of 
\pACPr, for the equivalence relation
$\set{\tup{t_1 \diamond t_2,t_1' \diamond t_2'} \where
      R_1(t_1,t_1') \Land R_2(t_2,t_2')}$.

Let $t_1,t_1',t_2,t_2'$ be closed \pACPr\ terms such that 
$t_1 \bisim t_1'$ and $t_2 \bisim t_2'$, and 
let $R_1$ and $R_2$ be probabilistic bisimulations witnessing 
$t_1 \bisim t_1'$ and $t_2 \bisim t_2'$, respectively.

For each binary operator $\diamond$ of \pACPr, we construct an 
equivalence relation $R_\diamond$ on closed \pACPr\ terms as follows:
\begin{ldispl}
\begin{ceqns}
\text{in the case that } 
\diamond \text{ is } \seqc : &
R_\diamond = ((R_1 \diamond R_2) \union R_2)\eqvcl\;; 
\\
\text{in the case that } 
\diamond \text{ is } \altc, \paltc{\pi} \text{ or } \parc : &
R_\diamond = ((R_1 \diamond R_2) \union R_1 \union R_2)\eqvcl\;; 
\\
\text{in the case that } 
\diamond \text{ is } \leftm \text{ or } \commm : &
R_\diamond = ((R_1 \diamond R_2) \union (R_1 \parc R_2) \union
     R_1 \union R_2)\eqvcl
\end{ceqns}
\end{ldispl}%
and for each encapsulation operator $\encap{H}$, we construct an 
equivalence relation $R_\encap{H}$ on closed \pACPr\ terms as follows: 
\begin{ldispl}
R_\encap{H} =  
(\set{\tup{\encap{H}(t_1),\encap{H}(t_1')} \where R_1(t_1,t_1')} \union
 R_1)\eqvcl\;.
\end{ldispl}%

For each operator $\diamond$ of \pACPr, we have to show that the 
conditions from the definition of a probabilistic bisimulation hold for 
the constructed relation~$R_\diamond$.

The proofs that the conditions concerning the relations $\astep{}{a}{}$ 
and $\aterm{}{a}$ hold are easy.
The proof that the condition concerning the function $P$ holds is  
straightforward using the property of $P$ mentioned in the proof of
Proposition~\ref{prop-bisim-max} and the following easy-to-check 
properties of $P$:
\begin{ldispl}
\begin{aceqns}
P(t \seqc t',T \seqc T') & = & 0      & \text{if } t' \notin T'\;, \\
P(t \seqc t',T \seqc T') & = & P(t,T) & \text{if } t' \in T'\;,    \\
P(t \altc t',T \altc T') & = & P(t,T) \mul P(t',T')\;, \\
P(t \paltc{\pi} t',T) & = & \pi \mul P(t,T) + (1-\pi) \mul P(t',T)\;, \\
P(t \parc t',T \parc T')   & = & P(t,T) \mul P(t',T')\;, \\
P(t \leftm t',T \leftm T') & = & P(t,T) \mul P(t',T')\;, \\
P(t \commm t',T \commm T') & = & P(t,T) \mul P(t',T')\;, \\
P(\encap{H}(t),\encap{H}(T)) & = & P(t,T)\;,
\end{aceqns}
\end{ldispl}%
where we write $T \diamond T'$, where $T$ and $T'$ are sets of closed
\pACPr\ terms and $\diamond$ is a binary operator of \pACPr, for the 
set $\set{t \diamond t' \where t \in T \Land t' \in T'}$ and we write
$\encap{H}(T)$, where $T$ is a set of closed \pACPr\ terms, for the set 
$\set{\encap{H}(t) \where t \in T}$.
\qed
\end{proof}
\pACPpa\ is the variant of \pACP\ with a different parallel composition 
operator that is presented in~\cite{And99a,And02a}.%
\footnote{\pACPpa\ is called \pACPp\ in~\cite{And99a}.}
A detailed proof of Proposition~\ref{prop-bisim-congr} is to a 
large extent a simplified version of the detailed proof of the fact that 
${} \bisim {}$ is a congruence with respect to the operators of \pACPpa\ 
that is given in~\cite{And02a}.
This is because of the fact that, except for the parallel composition 
operator, the structural operational semantics of \pACP\ presented in 
this paper can essentially be obtained from the structural operational 
semantics of \pACPpa\ that is presented in~\cite{And02a} by removing
unnecessary complexity.

In~\cite{LT09a}, constraints have been proposed on the form of 
operational semantics rules which ensure that probabilistic bisimulation 
equivalence is a congruence.
Both the reactive and generative models of probabilistic processes 
(see~\cite{GSS95a}) are covered in that paper.
While \pACPr\ is based on the generative model, virtually all other work 
in this area covers the reactive model only.
Unfortunately, the relations used for the structural operational 
semantics of \pACPr\ differ from the ones used in~\cite{LT09a}.
The chances are that the structural operational semantics of \pACPr\ can 
be adapted such that the results from that paper can be used to prove
Proposition~\ref{prop-bisim-congr}.
Howver, it seems quite likely that such a proof requires much more 
effort than the proof sketched above.

\pACPr\ is sound with respect to probabilistic bisimulation equivalence 
for equations between closed terms.
\begin{theorem}[Soundness]
\label{thm-soundness-pACPr}
For all closed \pACPr\ terms $t$ and $t'$, $t = t'$ is derivable from 
the axioms of \pACPr\ only if $t \bisim t'$.
\end{theorem}
\begin{proof}
Since ${} \bisim {}$ is a congruence for \pACPr, we only need to verify 
the soundness of each axiom of \pACPr.

For each equational axiom $e$ of \pACPr\ (all axioms of \pACPr\ except 
$\axiom{CM1'}$ and RSP are equational), we construct an equivalence 
relation $R_e$ on closed \pACPr\ terms as follows:
\begin{ldispl}
R_e = 
{\set{\tup{t,t'} \where 
      t = t' \text{ is a closed substitution instance of } e}}\eqvcl\;.
\end{ldispl}%

For axiom $\axiom{CM1'}$, we construct an equivalence relation $R'$ on 
closed \pACPr\ terms as follows:
\begin{ldispl}
R' = 
\set{\tup{t,t'} \where 
     t = t' \text{ is a closed substitution instance of } e \Land
     \pstep{t}{1}{t} \Land \pstep{t'}{1}{t'}}\eqvcl\;,
\end{ldispl}%
where $e$ is the consequent of $\axiom{CM1'}$.

For an arbitrary instance 
$\set{X_i = t_i \where i \in I} \Limpl
 X_j = \rec{X_j}{\set{X_i = t_i \where i \in I}}$
of RSP ($j \in I$), we construct an equivalence relation $R''$ on closed 
\pACPr\ terms as follows:
\begin{ldispl}
R'' = 
\set{\tup{\theta(X_j),\rec{X_j}{\set{X_i = t_i \where i \in I}}} \where 
     j \in I \Land \theta \in \Theta \Land
     \LAND{i \in I} \theta(X_i) \bisim \theta(t_i)}\eqvcl\;,
\end{ldispl}%
where 
$\Theta$ is the set of all functions from $\cX$ to the set of all closed 
\pACPr\ terms and 
$\theta(t)$, where $\theta \in \Theta$ and $t$ is a \pACPr\ term, stands
for $t$ with, for all $X \in \cX$, all occurrences of $X$ replaced by 
$\theta(X)$.

For each equational axiom $e$ of \pACPr, we have to check whether the 
conditions from the definition of a probabilistic bisimulation hold for 
the constructed relation $R_e$.
For axiom $\axiom{CM1'}$, we have to check whether the conditions from 
the definition of a probabilistic bisimulation hold for the constructed 
relation $R'$.
That this is sufficient for the soundness of axiom $\axiom{CM1'}$
follows from Proposition~\ref{prop-bisim-pstep-one}.
For the instances of axiom RSP, we have to check whether the conditions 
from the definition of a probabilistic bisimulation hold for the 
constructed relation $R''$.

All these checks are straightforward, for the condition concerning the 
function $P$, using the following easy-to-check property of $P$:
if $\beta$ is a bijection on $T$ and $P(t',t) = P(t'',\beta(t))$ for all 
$t \in T$, then $P(t',T) = P(t'',T)$.
\qed
\end{proof}
In versions of \ACP\ where RSP follows from RDP and AIP (Approximation
Induction Principle), soundness of RSP follows from soundness of RDP 
and AIP (see e.g.~\cite{BW90}).

The following three lemmas will be used below in the proof of a
completeness theorem for \pACP.
For convenience, we introduce the notion of a rigid closed \pACP\ term.

A closed \pACP\ term $t$ is \emph{rigid} if, for all probabilistic 
bisimulations $R$, $R(t,t)$ only if the restriction of $R$ to the set of 
all subterms of $t$ is the identity relation on that set.

\begin{lemma}
\label{lemma-basic-rigid}
All proper basic terms $t$ of \pACP\ are rigid.
\end{lemma}
\begin{proof}
This is easily proved by induction on the structure of $t$.
\qed
\end{proof}

\begin{lemma}
\label{lemma-rigid-bijection}
For all rigid closed \pACP\ terms $t$ and $t'$, for all probabilistic 
bisimulations $R$ with $R(t,t')$, the restriction of $R$ to the set of 
all subterms of $t$ is a bijection.
\end{lemma}
\begin{proof}
Suppose there exist subterms $t_1$ and $t_2$ of $t$ and a subterm $t''$
of $t'$ such that $R(t_1,t'')$ and $R(t_2,t'')$.
Because $R(t,t')$, $R^{-1}$ is a probabilistic bisimulation such that 
$R^{-1}(t',t)$ and $R^{-1} \circ R$ is a probabilistic bisimulation such
that $R^{-1} \circ R(t,t)$.
We also have that $R^{-1} \circ R(t_1,t_2)$. 
Because $t$ is rigid, it follows that $t_1 = t_2$.
\qed
\end{proof}

\begin{lemma}
\label{lemma-bijection-deriv}
For all proper basic term $t$ and $t'$ of \pACP, there exists a 
probabilistic bisimulation $R$ with $R(t,t')$ such that the restriction 
of $R$ to the set of all subterms of $t$ is a bijection only if $t = t'$ 
is derivable from axioms A1, A2, pA1, and pA2.
\end{lemma}
\begin{proof}
This is straightforwardly proved by induction on the structure of $t$.
\qed
\end{proof}

\begin{theorem}[Completeness]
\label{thm-completeness-pACP}
For all closed \pACP\ terms $t$ and $t'$, $t \bisim t'$ only if 
$t = t'$ is derivable from the axioms of \pACP.
\end{theorem}
\begin{proof}
By Proposition~\ref{prop-elim-basic} and 
Theorem~\ref{thm-soundness-pACPr}, it is sufficient to prove the 
theorem for proper basic terms $t$ and $t'$ of \pACP.
Assume that $t \bisim t'$.
Then, there exists a probabilistic bisimulation $R$ such that $R(t,t')$.
By Lemma~\ref{lemma-basic-rigid}, $t$ and $t'$ are rigid.
So, by Lemma~\ref{lemma-rigid-bijection}, the restriction of $R$ to the 
set of all subterms of $t$ is a bijection.
From this, by Lemma~\ref{lemma-bijection-deriv}, it follows that 
$t = t'$ is derivable from axioms A1, A2, pA1, and pA2.
\qed
\end{proof}

\subsection{Remarks Relating to the Semantics of \pACPr}
\label{subsect-remarks}

In this subsection, we make some remarks, relating to the operational 
semantics of \pACPr, that did not fit in very well at an earlier point.

\pACP\ is a minor variant of the subtheory of \pACPt\ from~\cite{AG09a} 
in which the operators for abstraction from some set of actions are 
lacking.
Soundness and completeness results with respect to branching 
bisimulation equivalence of an unspecified operational semantics of 
\pACPt\ are claimed in~\cite{AG09a}.
In principle, the operational semantics concerned should be derivable 
from the operational semantics of \pTCPt\ given in~\cite{Geo11a}.%
\footnote
{Recall that \pTCPt\ is \pACPt\ with the action constants replaced by
 action prefixing operators and a constant for the process that is only
 capable of terminating successfully.}
However, it turns out that a mistake has been made in the rules for the 
probabilistic choice operators that concern the relations 
$\pstep{}{\pi}{}$.
The mistake concerned manifests only in closed terms of the form 
$t \paltc{1/2} t$.  
For example, if $t$ is not derivably equal to a term whose outermost 
operator is a probabilistic choice operator, then both the left-hand 
side and the right-hand side of $t \paltc{1/2} t$ give rise to 
% !!
\smash{$\pstep{t \paltc{1/2} t}{1/2}{t}$}.
Consequently, the total probability that $t \paltc{1/2} t$ behaves as 
$t$ is $1/2$ instead of~$1$.
This is counterintuitive and inconsistent with axiom pA3.

A meadow has a total multiplicative inverse operation where the 
multiplicative inverse of zero is zero.
This is why there is no reason to exclude the probabilistic choice 
operators $\paltc{\pi}$ for $\pi \in \set{0,1}$ if a meadow is used 
instead of a field.
Because we have included these operators, we also have included 
relations $\pstep{}{\pi}{}$ for $\pi \in \set{0,1}$.
As a bonus of the inclusion of these relations, we could achieve that 
for all pairs $(t,t')$ of closed \pACPr\ terms, there exists a 
$\pi \in \Prob$ such that $\pstep{t}{\pi}{t'}$.
Due to this, we could at the same time reduce the number of rules for 
the operational semantics that concern the relations $\pstep{}{\pi}{}$,
replace all negative premises by positive premises in rules for the 
operational semantics that concern the relations $\astep{}{a}{}$ and
$\aterm{}{a}$, and correct the above-mentioned mistake in the rules for 
the probabilistic choice operators that concern the relations 
$\pstep{}{\pi}{}$.

Above, we already mentioned that a variant of \pACP, called \pACPpa, is 
presented in~\cite{And99a,And02a}.
\pACP, just like \pACPt\ from~\cite{AG09a}, differs from \pACPpa\ with 
respect to the parallel composition operator.
Moreover, in~\cite{And99a,And02a}, the probability distribution function 
is defined directly instead of via the operational semantics.
However, except for parallel composition and left merge, the probability 
distribution function corresponds to the probability distribution 
function $P$ defined above.
The direct definition of the probability distribution function removes 
the root of the above-mentioned mistake made in~\cite{Geo11a}. 

\section{Probabilistic Strategic Interleaving}
\label{sect-pSI}

In this section, we extend \pACP\ with probabilistic strategic 
interleaving, i.e.\ interleaving according to some probabilistic 
interleaving strategy.
Interleaving strategies are known as process-scheduling policies in the 
field of operating systems.
A well-known probabilistic process-scheduling policy is lottery 
scheduling~\cite{WW94a}.
In the presented extension of \pACP\ deterministic interleaving 
strategies are special cases of probabilistic interleaving strategies: 
they are the ones obtained by restriction to the trivial probabilities 
$0$ and $1$.

\subsection{Motivation for Strategic Interleaving}
\label{subsect-motivation}

In this subsection, the motivation for taking strategic interleaving 
into consideration is given.

The interest in strategic interleaving originates from an important 
feature of many contemporary programming languages, namely 
multi-threading.
In algebraic theories of processes, such as \ACP~\cite{BW90}, 
CCS~\cite{Mil89}, and CSP~\cite{Hoa85}, processes are discrete 
behaviours that proceed by doing steps in a sequential fashion.
In these theories, parallel composition of two processes is usually
interpreted as arbitrary interleaving of the steps of the processes 
concerned. 
Arbitrary interleaving turns out to be appropriate for many applications 
and to facilitate formal algebraic reasoning. 
Multi-threading as found in programming languages such as 
Java~\cite{GJSB00a} and C\#~\cite{HWG03a}, gives rise to parallel 
composition of processes.
In the case of multi-threading, however, the steps of the processes 
concerned are interleaved according to what is known as a 
process-scheduling policy in the field of operating systems.

Arbitrary interleaving and strategic interleaving are quite different.
The following points illustrate this: (a) whether the interleaving of 
certain processes leads to inactiveness depends on the interleaving 
strategy used; (b) sometimes inactiveness occurs with a particular 
interleaving strategy whereas arbitrary interleaving would not lead to 
inactiveness and vice versa.
Nowadays, multi-threading is often used in the implementation of 
systems.
Because of this, in many systems, for instance hardware/software 
systems, we have to do with parallel processes that may best be 
considered to be interleaved in an arbitrary way as well as parallel 
processes that may best be considered to be interleaved according to 
some interleaving strategy.
Such applications potentially ask for a process algebra that supports 
both arbitrary interleaving and strategic interleaving.

\subsection{\pACP\ with Probabilistic Strategic Interleaving}
\label{subsect-pSI}

In the extension of \pACP\ with probabilistic strategic interleaving 
presented below, it is expected that an interleaving strategy uses the 
interleaving history in one way or another to make process-scheduling 
decisions.

The sets $\Hist_n$ of \emph{interleaving histories for $n$ processes}, 
for $n \in \Natpos$, are the subsets of $\seqof{(\Natpos \x \Natpos)}$ 
that are inductively defined by the following rules:%
\footnote
{The special sequence notation used in this paper is explained in an
 appendix.}
\begin{itemize}
\item
$\emptyseq \in \Hist_n$;
\item
if $i \leq n$, then $\tup{i,n} \in \Hist_n$;
\item
if $h \concat \tup{i,n} \in \Hist_n$, $j \leq n$, and 
$n - 1 \leq m \leq n + 1$, then 
$h \concat \tup{i,n} \concat \tup{j,m} \in \Hist_m$.
\end{itemize}
The intuition concerning interleaving histories is as follows:
if the $k$th pair of an interleaving history is $\tup{i,n}$, then the 
$i$th process got a turn in the $k$th interleaving step and after its
turn there were $n$ processes to be interleaved.
The number of processes to be interleaved may increase due to process
creation (introduced below) and decrease due to successful termination 
of processes.
 
The presented extension of \pACP\ is called \pACPpSI\ (\pACP\ with 
probabilistic Strategic Interleaving). 
It covers a generic probabilistic interleaving strategy that can be 
instantiated with different specific probabilistic interleaving 
strategies that can be represented in the way that is explained below.

In \pACPpSI, it is assumed that the following has been given:% 
\footnote
{We write $\pfunct{f}{A}{B}$ to indicate that $f$ is a partial function
 from $A$ to $B$.}

\begin{itemize}
\item
a fixed but arbitrary set $S$; 
\item
a fixed but arbitrary partial function 
$\pfunct{\sched{n}}{\Hist_n \x S}{(\set{1,\ldots,n} \to \Prob)}$ 
for each \linebreak[2] $n \in \Natpos$;
\item
a fixed but arbitrary total function 
$\funct{\updat{n}}
  {\Hist_n \x S \x \set{1,\ldots,n} \x \Act \x \set{0,1}}{S}$ 
for \linebreak[2] each $n \in \Natpos$;
\item
a fixed but arbitrary set $C \subset \Act$;
\end{itemize}
where, for each $n \in \Natpos$:
\begin{itemize}
\item
for each $h \in \Hist_n$ and $s \in S$, 
$\sum_{i=1}^n \sched{n}(h,s)(i) = 1$;
\item
for each $h \in \Hist_n$, $s \in S$, $i \in \set{1,\ldots,n}$, and
$a \in \Act \diff C$, $\updat{n}(h,s,i,a,0) = s$;
\item
for each $c \in C$, $\ol{c} \in \Act \diff C$ and, 
for each $a,b \in \Act$, 
$\commf(a,b) \neq c$, $\commf(a,b) \neq \ol{c}$,
$\commf(a,c) = \dead$, and $\commf(a,\ol{c}) = \dead$.
\end{itemize}
The elements of $S$ are called \emph{control states}, $\sched{n}$ is 
called an \emph{abstract scheduler} (\emph{for $n$ processes}),  
$\updat{n}$ is called a \emph{control state transformer} (\emph{for 
$n$ processes}), and 
the elements of $C$ are called \emph{control actions}.
The intuition concerning $S$, $\sched{n}$, $\updat{n}$, and $C$ is as 
follows:
\begin{itemize}
\item
the control states from $S$ encode data that are relevant to the 
interleaving strategy, but not derivable from the interleaving history;
\item
if $\sched{n}(h,s) = i$, then the $i$th process gets the next turn after 
interleaving history $h$ in control state $s$;
\item
if $\sched{n}(h,s)$ is undefined, then no process gets the next turn 
after interleaving history $h$ in control state $s$; 
\item
if $\updat{n}(h,s,i,a,0) = s'$, then $s'$ is the control state that 
arises from the $i$th process doing $a$ after interleaving 
history $h$ in control state $s$ in the case that doing $a$ does not 
bring the $i$th process to successful termination;
\item
if $\updat{n}(h,s,i,a,1) = s'$, then $s'$ is the control state that 
arises from the $i$th process doing $a$ after interleaving 
history $h$ in control state $s$ in the case that doing $a$ brings the 
$i$th process to successful termination;
\item
if $a \in C$, then $a$ is an explicit means to bring about a control 
state change and $\ol{a}$ is left as a trace after $a$ has been dealt 
with. 
\end{itemize}
Thus, $S$, $\indfam{\sched{n}}{n \in \Natpos}$,  
$\indfam{\updat{n}}{n \in \Natpos}$, and $C$ together represent an 
interleaving strategy.
This way of representing an interleaving strategy is engrafted 
on~\cite{SS00a}.

Consider the case where $S$ is a singleton set, 
for each $n \in \Natpos$, $\sched{n}$ is defined~by 
\begin{ldispl}
\begin{gceqns}
\sched{n}(\emptyseq,s)(i) = 1\;
 & \mathrm{if} \;i = 1\;,  \\
\sched{n}(\emptyseq,s)(i) = 0\;
 & \mathrm{if} \;i \neq 1\;, \\
\sched{n}(h \concat \tup{j,n},s)(i) = 1\;
 & \mathrm{if} \;i = (j \bmod n) + 1\;, \\
\sched{n}(h \concat \tup{j,n},s)(i) = 0\;
 & \mathrm{if} \;i \neq (j \bmod n) + 1
\end{gceqns}
\end{ldispl}%
and, for each $n \in \Natpos$, $\updat{n}$ is defined by 
\begin{ldispl}
\updat{n}(h,s,i,a,f) = s\;.
\end{ldispl}%
In this case, the interleaving strategy corresponds to the round-robin 
scheduling algorithm.
This deterministic interleaving strategy is called cyclic interleaving 
in our work on interleaving strategies in the setting of thread algebra 
(see e.g.~\cite{BM04c}).
In the current setting, an interleaving strategy is deterministic if, 
for all $n \in \Natpos$, for all $h \in \Hist_n$, $s \in S$, and 
$i \in \set{1,\ldots,n}$, $\sched{n}(h,s)(i) \in \set{0,1}$.
In the case that $S$ and $\updat{n}$ are as above, but $\sched{n}$ is 
defined by
\begin{ldispl}
\sched{n}(h,s)(i) = 1 / n\;,
\end{ldispl}%
the interleaving strategy is a purely probabilistic one.
The probability distribution used is a uniform distribution.

More advanced strategies can be obtained if the scheduling makes more 
advanced use of the interleaving history and the control state.
The interleaving history may, for example, be used to factor the 
individual lifetimes of the processes to be interleaved or their 
creation hierarchy into the process-scheduling decision making.
Individual properties of the processes to be interleaved that depend on 
actions performed by them can be taken into account by making use of the 
control state.
The control state may, for example, be used to factor whether a process 
is currently waiting to acquire a lock from a process that manages a 
shared resource into the process-scheduling decision making.
An example of a probabilistic interleaving strategy supporting mutual 
exclusion of critical subprocesses is given in 
Section~\ref{subsect-mutex}.

In \pACPpSI, it is also assumed that a fixed but arbitrary set $D$ of 
\emph{data} and a fixed but arbitrary function $\funct{\crea}{D}{P}$, 
where $P$ is the set of all closed terms over the signature of \pACPpSI\ 
(given below), have been given and that, for each $d \in D$ and 
$a, b \in \Act$, $\pcr(d),\rcr(d) \in \Act$, 
$\commf(\pcr(d),a) = \dead$, and $\commf(a,b) \neq \pcr(d)$.
The action $\pcr(d)$ can be considered a process creation request and 
the action $\rcr(d)$ can be considered a process creation act.
They represent the request to start the process denoted by $\crea(d)$ in 
parallel with the requesting process and the act of carrying out that 
request, respectively.

The signature of \pACPpSI\ consists of the constants and operators
from the signature of \pACP\ and in addition the following operators:
\pagebreak[2]
\begin{itemize}
\item
for each $n \in \Natpos$, $h \in \Hist_n$, and $s \in S$,
the $n$-ary \emph{strategic interleaving} operator $\siop{n}{h}{s}$;
\item
for each $n,i \in \Natpos$ with $i \leq n$, $h \in \Hist_n$, and 
$s \in S$,
the $n$-ary \emph{positional strategic interleaving} operator
$\posmop{n}{i}{h}{s}$.
\end{itemize}

The strategic interleaving operators can be explained as follows:
\begin{itemize}
\item
a closed term of the form $\si{n}{h}{s}{t_1,\ldots,t_n}$ denotes the 
process that results from interleaving of the $n$ processes denoted by 
$t_1,\ldots,t_n$ after interleaving history $h$ in control state $s$, 
according to the interleaving strategy represented by $S$, 
$\indfam{\sched{n}}{n \in \Natpos}$, and 
$\indfam{\updat{n}}{n \in \Natpos}$.
\end{itemize}
The positional strategic interleaving operators are auxiliary operators 
used to axiomatize the strategic interleaving operators.
The role of the positional strategic interleaving operators in the 
axiomatization is similar to the role of the left merge operator found 
in \pACP.

The axioms of \pACPpSI\ are the axioms of \pACP\ and in addition the 
equations given in Table~\ref{axioms-pSI}.
\begin{table}[!t]
\caption{Axioms for strategic interleaving}
\label{axioms-pSI}
\begin{eqntbl}
\begin{axcol}
x_1 = x_1 \altc x_1 \Land \ldots \Land x_1 = x_n \altc x_n \Limpl 
\\ \quad
\si{n}{h}{s}{x_1,\ldots,x_n} = \dead
\hfill \mif \sched{n}(h,s) \mathrm{\;is\;undefined}   & \axiom{SI0'} \\
x_1 = x_1 \altc x_1 \Land \ldots \Land x_1 = x_n \altc x_n \Limpl 
\\ \quad
\si{n}{h}{s}{x_1,\ldots,x_n} = 
\PRC{i = 1}{n}{\sched{n}(h,s)(i)}{\posm{n}{i}{h}{s}{x_1,\ldots,x_n}}
\hfill\;\;
       \mif \sched{n}(h,s) \mathrm{\;is\;defined}\phantom{\mathrm{un}}       
                                                      & \axiom{SI1'} \\
\posm{n}{i}{h}{s}{x_1,\ldots,x_{i-1},\dead,x_{i+1},\ldots,x_n} = \dead
                                                      & \axiom{SI2}  \\
\posm{1}{i}{h}{s}{a} = a                              & \axiom{SI3}  \\
\posm{n+1}{i}{h}{s}{x_1,\ldots,x_{i-1},a,x_{i+1},\ldots,x_{n+1}} =
\\ \quad
a \seqc
\si{n}{h \concat \tup{i,n}}{\updat{n+1}(h,s,i,a,1)}
 {x_1,\ldots,x_{i-1},x_{i+1},\ldots,x_{n+1}}          & \axiom{SI4}  \\
\posm{n}{i}{h}{s}{x_1,\ldots,x_{i-1},a \seqc x_i',x_{i+1},\ldots,x_n} =
\\ \quad
a \seqc
\si{n}{h \concat \tup{i,n}}{\updat{n}(h,s,i,a,0)}
 {x_1,\ldots,x_{i-1},x_i',x_{i+1},\ldots,x_n}         & \axiom{SI5}  \\
\posm{n}{i}{h}{s}{x_1,\ldots,x_{i-1},\pcr(d),x_{i+1},\ldots,x_n} =
\\ \quad
\rcr(d) \seqc
\si{n}{h \concat \tup{i,n}}{\updat{n}(h,s,i,\pcr(d),1)}
 {x_1,\ldots,x_{i-1},x_{i+1},\ldots,x_n,\crea(d)}     & \axiom{SI6}  \\
\posm{n}{i}{h}{s}
 {x_1,\ldots,x_{i-1},\pcr(d) \seqc x_i',x_{i+1},\ldots,x_n} =
\\ \quad
\rcr(d) \seqc
\si{n+1}{h \concat \tup{i,n+1}}{\updat{n}(h,s,i,\pcr(d),0)}
 {x_1,\ldots,x_{i-1},x_i',x_{i+1},\ldots,x_n,\crea(d)} 
                                                      & \axiom{SI7}  \\ 
\posm{n}{i}{h}{s}
 {x_1,\ldots,x_{i-1},x_i' \altc x_i'',x_{i+1},\ldots,x_n} =
\\ \quad
\posm{n}{i}{h}{s}{x_1,\ldots,x_{i-1},x_i',x_{i+1},\ldots,x_n} \altc
\posm{n}{i}{h}{s}{x_1,\ldots,x_{i-1},x_i'',x_{i+1},\ldots,x_n} 
                                                      & \axiom{SI8}  \\ 
\multicolumn{2}{@{}c@{}}{\dotfill} \\[1ex] 
\si{n}{h}{s}
 {x_1,\ldots,x_{i-1},x_i' \paltc{\pi} x_i'',x_{i+1},\ldots,x_n} =
\\ \quad
\si{n}{h}{s}{x_1,\ldots,x_{i-1},x_i',x_{i+1},\ldots,x_n}
 \paltc{\pi}
\si{n}{h}{s}{x_1,\ldots,x_{i-1},x_i'',x_{i+1},\ldots,x_n} 
                                                      & \axiom{pSI1} \\
\posm{n}{i}{h}{s}
 {x_1,\ldots,x_{i-1},x_i' \paltc{\pi} x_i'',x_{i+1},\ldots,x_n} =
\\ \quad
\posm{n}{i}{h}{s}{x_1,\ldots,x_{i-1},x_i',x_{i+1},\ldots,x_n}
 \paltc{\pi}
\posm{n}{i}{h}{s}{x_1,\ldots,x_{i-1},x_i'',x_{i+1},\ldots,x_n} 
                                                      & \axiom{pSI2} 
\end{axcol}
\end{eqntbl}
\end{table}
In the additional equations, 
$n$ and $i$ stand for arbitrary numbers from $\Natpos$, 
$h$ stands for an arbitrary interleaving history from $\Hist$, 
$s$ stands for an arbitrary control state from $S$, 
$a$ stands for an arbitrary action constant that is not of the form 
$\pcr(d)$ or $\rcr(d)$, and 
$d$ stands for an arbitrary datum $d$ from~$D$.

The equations in Table~\ref{axioms-pSI} above the dotted line are 
similar to the axioms for strategic interleaving presented 
in~\cite{BM17b} for the deterministic case.
The difference between \axiom{SI1} from that paper and the consequent 
of \axiom{SI1'} is unavoidable because probabilistic interleaving 
strategies are not covered there.
The other differences are due to the finding that the generic 
interleaving strategy from~\cite{BM17b} cannot be instantiated with: 
(a)~interleaving strategies where the data relevant to the 
process-scheduling decision making may be such that none of the 
processes concerned can be given a turn,
(b)~interleaving strategies where the data relevant to the 
process-scheduling decision making must be updated on successful 
termination of one of the processes concerned, and 
(c)~interleaving strategies where the process-scheduling decision
making may be adjusted by steps of the processes concerned that are 
solely intended to change the data relevant to the process-scheduling 
decision making.

Axiom SI2 expresses that, in the event of inactiveness of the process 
whose turn it is, the whole becomes inactive immediately.
A plausible alternative is that, in the event of inactiveness of the 
process whose turn it is, the whole becomes inactive only after all 
other processes have terminated or become inactive.
In that case, the functions 
$\funct{\updat{n}}
  {\Hist \x S \x \set{1,\ldots,n} \x \Act \x \set{0,1}}{S}$ 
must be extended to functions 
$\funct{\updat{n}}
 {\Hist \x S \x \set{1,\ldots,n} \x (\Act \union \set{\dead}) \x
  \set{0,1}}{S}$
and axiom SI2 must be replaced by the axioms in 
Table~\ref{axioms-alt-inactive}.
\begin{table}[!t]
\caption{Alternative axioms for SI2}
\label{axioms-alt-inactive}
\begin{eqntbl}
\begin{axcol}
\posm{1}{i}{h}{s}{\dead} = \dead                       & \axiom{SI2a} \\
\posm{n+1}{i}{h}{s}{x_1,\ldots,x_{i-1},\dead,x_{i+1},\ldots,x_{n+1}} = 
\\ \qquad
\si{n}{h \concat \tup{i,n}}{\updat{n+1}(h,s,i,\dead,0)}
      {x_1,\ldots,x_{i-1},x_{i+1},\ldots,x_{n+1}} \seqc \dead         
                                                       & \axiom{SI2b} 
\end{axcol}
\end{eqntbl}
\end{table}

In \pACPpSIr, i.e.\ \pACPpSI\ extended with guarded recursion in the way 
described in Section~\ref{subsect-recursion}, the processes that can be 
created are restricted to the ones denotable by a closed \pACPpSI\ term.
This restriction stems from the requirement that $\crea$ is a function 
from $D$ to the set of all closed \pACPpSI\ terms.
The restriction can be removed by relaxing this requirement to the 
requirement that $\crea$ is a function from $D$ to the set of all closed 
\pACPpSIr\ terms. 
We write \pACPpSIrp\ for the theory resulting from this relaxation.
In other words, \pACPpSIrp\ differs from \pACPpSIr\ in that it is assumed 
that a fixed but arbitrary function $\funct{\crea}{D}{P}$, where $P$ is 
the set of all closed terms over the signature of \pACPpSIr, has been 
given.

\subsection{Semantics of \pACPpSI\ with Guarded Recursion}
\label{subsect-semantics-pSI}

In this subsection, we present a structural operational semantics of 
\pACPpSI\ with guarded recursion.

The structural operational semantics of \pACPpSIrp\ is described by the
rules for the operational semantics of \pACPr\ (given in 
Tables~\ref{rules-pACPr-1} and~\ref{rules-pACPr-2}) and in addition the 
rules given in Table~\ref{rules-pSI}.
\begin{table}[!t]
\caption{Additional rules for the operational semantics of \pACPpSIrp}
\label{rules-pSI}
\begin{ruletbl}
\Rule
{\aterm{x}{a}}
{\aterm{\posm{1}{1}{h}{s}{x}}{a}}
\\
\Rule
{\pstep{x_1}{1}{x_1'},\; \ldots,\; \pstep{x_{i-1}}{1}{x_{i-1}'},\; 
 \aterm{x_i}{a},\;
 \pstep{x_{i+1}}{1}{x_{i+1}'},\; \ldots,\; \pstep{x_{n+1}}{1}{x_{n+1}'}}
{\astep{\posm{n+1}{i}{h}{s}{x_1,\ldots,x_{n+1}}}{a}
 {\si{n}{h \concat \tup{i,n}}{\updat{n+1}(h,s,i,a,1)}
   {x_1,\ldots,x_{i-1},x_{i+1},\ldots,x_{n+1}}}}
\\
\Rule
{\pstep{x_1}{1}{x_1'},\; \ldots,\; \pstep{x_{i-1}}{1}{x_{i-1}'},\; 
 \astep{x_i}{a}{x'_i}, \;
 \pstep{x_{i+1}}{1}{x_{i+1}'},\; \ldots,\; \pstep{x_{n}}{1}{x_{n}'}}
{\astep{\posm{n}{i}{h}{s}{x_1,\ldots,x_n}}{a}
 {\si{n}{h \concat \tup{i,n}}{\updat{n}(h,s,i,a,0)}
   {x_1,\ldots,x_{i-1},x_i',x_{i+1},\ldots,x_n}}}
\\
\Rule
{\pstep{x_1}{1}{x_1'},\; \ldots,\; \pstep{x_{i-1}}{1}{x_{i-1}'},\; 
 \aterm{x_i}{\pcr(d)}, \;
 \pstep{x_{i+1}}{1}{x_{i+1}'},\; \ldots,\; \pstep{x_{n}}{1}{x_{n}'}}
{\astep{\posm{n}{i}{h}{s}{x_1,\ldots,x_n}}{\rcr(d)}
 {\si{n}{h \concat \tup{i,n}}{\updat{n}(h,s,i,\pcr(d),1)}
   {x_1,\ldots,x_{i-1},x_{i+1},\ldots,x_n,\crea(d)}}}
\\
\Rule
{\pstep{x_1}{1}{x_1'},\; \ldots,\; \pstep{x_{i-1}}{1}{x_{i-1}'},\; 
 \astep{x_i}{\pcr(d)}{x'_i}, \;
 \pstep{x_{i+1}}{1}{x_{i+1}'},\; \ldots,\; \pstep{x_{n}}{1}{x_{n}'}}
{\astep{\posm{n}{i}{h}{s}{x_1,\ldots,x_n}}{\rcr(d)}
 {\si{n+1}{h \concat \tup{i,n+1}}{\updat{n}(h,s,i,\pcr(d),0)}
   {x_1,\ldots,x_{i-1},x_i',x_{i+1},\ldots,x_n,\crea(d)}}}
\\
\RuleC
{\pstep{x_1}{\pi_1}{x_1'},\; \ldots,\; \pstep{x_{n}}{\pi_n}{x_{n}'}}
{\pstep{\si{n}{h}{s}{x_1,\ldots,x_{n}}}
 {\sched{n}(h,s)(i) \mul \pi_1 \mul \ldots \mul \pi_n}
 {\posm{n}{i}{h}{s}{x_1',\ldots,x_{n}'}}}
{\sched{n}(h,s) \mathrm{\;is\;defined}}
\\
\Rule
{\pstep{x_1}{\pi_1}{x_1'},\;\ldots,\;\pstep{x_{n}}{\pi_n}{x_{n}'}}
{\pstep{\posm{n}{i}{h}{s}{x_1,\ldots,x_{n}}}
 {\pi_1 \mul \ldots \mul \pi_n}
 {\posm{n}{i}{h}{s}{x_1',\ldots,x_{n}'}}}
\end{ruletbl}
\end{table}
In the additional rules, 
$n$ and $i$ stand for arbitrary numbers from $\Natpos$, 
$h$ stands for an arbitrary interleaving history from $\Hist$, 
$s$ stands for an arbitrary control state from $S$, 
$a$ stands for an arbitrary action from $\Act$ that is not of the form 
$\pcr(d)$ or $\rcr(d)$, 
$d$ stands for an arbitrary datum $d$ from $D$, and
$\pi_1$, \ldots, $\pi_n$ stand for arbitrary probabilities from~$\Prob$. 

\begin{proposition}
\label{prop-bisim-congr-pSI}
${} \bisim {}$ is a congruence w.r.t.\ the operators of \pACPpSIrp.
\end{proposition}
\begin{proof}
The proof goes along the same line as the proof of 
Proposition~\ref{prop-bisim-congr} 
\qed
\end{proof}

\pACPpSIrp\ is sound with respect to probabilistic bisimulation 
equivalence for equations between closed terms.
\begin{theorem}[Soundness]
\label{thm-soundness-pACPpSIrp}
For all closed \pACPpSIrp\ terms $t$ and $t'$, $t = t'$ is derivable 
from the axioms of \pACPpSIrp\ only if $t \bisim t'$.
\end{theorem}
\begin{proof}
The proof goes along the same line as the proof of 
Theorem~\ref{thm-soundness-pACPr}.
\qed
\end{proof}

\subsection{Guarded Recursive Specifications over \pACP\ and \pACPpSI}
\label{subsect-expressivity}

In this subsection, we show that each guarded recursive specifications 
over \pACPpSI\ can be reduced to a guarded recursive specification over 
\pACP.
We make use of the fact that each guarded \pACPpSI\ term has a head 
normal form. 

Let $T$ be \pACPpSI\ or \pACPpSIr.
The set $\HNF$ of \emph{head normal forms of} $T$ is inductively defined 
by the following rules:
\begin{itemize}
\item 
$\dead \in \HNF$;
\item 
if $a \in \Act$, then $a \in \HNF$;
\item 
if $a \in \Act$ and $t$ is a $T$ term, then $a \seqc t \in \HNF$;
\item 
if $t,t' \in \HNF$, then $t \altc t '\in \HNF$;
\item 
if $t,t' \in \HNF$ and $\pi \in \Prob$, 
then $t \paltc{\pi} t '\in \HNF$.
\end{itemize}
Each head normal form of $T$ is derivably equal to a head normal 
form of the form $\PRC{i=1}{n}{\pi_i}{s_i}$, where $n \in \Natpos$ and, 
for each $i \in \Natpos$ with $i \leq n$, $s_i$ is of the form
$\Altc{j=1}{n_i} a_{ij} \seqc t_{ij} \altc \Altc{k=1}{m_i} b_{ik}$,
where $n_i,m_i \in \Natpos$ and, for all $j \in \Natpos$ with 
$j \leq n_i$, $a_{ij} \in \Act$ and $t_{ij}$ is a $T$ term, and, 
for all $k \in \Natpos$ with $k \leq m_i$, $b_{ik} \in \Act$. 

Each guarded \pACPpSIr\ term is derivably equal to a head normal form of 
\pACPpSIr.
\begin{proposition}
\label{prop-HNF-pACPpSIr}
For each guarded \pACPpSIr\ term $t$, there exists a head normal form 
$t'$ of \pACPpSIr\ such that $t = t'$ is derivable from the axioms of 
\pACPpSIr.
\end{proposition}
\begin{proof}
First we prove the following weaker result about head normal forms:
\begin{quote}
\emph{For each guarded \pACPpSI\ term $t$, there exists a head normal 
form $t'$ of \pACPpSI\ such that $t = t'$ is derivable from the axioms 
of \pACPpSI}.
\end{quote}
The proof is straightforward by induction on the structure of $t$.
The case where $t$ is of the form $\dead$ and the case where $t$ is of 
the form $a$ ($a \in \Act$) are trivial.
The case where $t$ is of the form $t_1 \seqc t_2$ follows immediately 
from the induction hypothesis (applied to $t_1$) and the claim that, for 
all head normal forms $t_1'$ and $t_2'$ of \pACPpSI, there exists a head 
normal form $t'$ of \pACPpSI\ such that $t_1' \seqc t_2' = t'$ is 
derivable from the axioms of \pACPpSI.
This claim is easily proved by induction on the structure of $t_1'$.
The cases where $t$ is of the form $t_1 \altc t_2$ or 
$t_1 \paltc{\pi} t_2$ follow immediately from the induction hypothesis.
The cases where $t$ is of one of the forms $t_1 \leftm t_2$, 
$t_1 \commm t_2$ or $\encap{H}(t_1)$ are proved in the same vein as 
the case where $t$ is of the form $t_1 \seqc t_2$.
In the case that $t$ is of the form $t_1 \commm t_2$, each of the cases 
to be considered in the inductive proof of the claim demands a (nested) 
proof by induction on the structure of $t_2'$.
The case that $t$ is of the form $t_1 \parc t_2$ follows immediately 
from the case that $t$ is of the form $t_1 \leftm t_2$ and the case that 
$t$ is of the form $t_1 \commm t_2$.
The case where $t$ is of the form $\posm{n}{i}{h}{s}{t_1,\ldots,t_n}$ is 
proved in the same vein as the case where $t$ is of the form 
$t_1 \seqc t_2$, but the claim is of course proved by induction on the 
structure of $t_i'$ instead of $t_1'$.
The case that $t$ is of the form $\si{n}{h}{s}{t_1,\ldots,t_n}$ follows 
immediately from the case that $t$ is of the form 
\smash{$\posm{n}{i}{h}{s}{t_1,\ldots,t_n}$}.
Because $t$ is a guarded \pACPpSI\ term, the case where $t$ is a 
variable cannot occur.

The proof of the proposition itself is also straightforward by induction 
on the structure of $t$.
The cases other than the case where $t$ is of the form $\rec{X}{E}$ is 
proved in the same way as in the above proof of the weaker result.
The case where $t$ is of the form $\rec{X}{E}$ follows immediately from
the weaker result and RDP.
\qed
\end{proof}

The following theorem refers to three process algebras.
It is implicit that the same set $\Act$ of actions and the same 
communication function $\commf$ are assumed in the process algebras
referred to.

Each guarded recursive specification over \pACPpSI\ can be reduced to a
guarded recursive specification over \pACP. 
\begin{theorem}[Expressivity]
\label{thm-expressivity}
\sloppy
For each guarded recursive specification $E$ over \pACPpSI\ and each 
$X \in \vars(E)$, there exists a guarded recursive specification $E'$ 
over \pACP\ such that $\rec{X}{E} = \rec{X}{E'}$ is derivable from the 
axioms of \pACPpSIr.
\end{theorem}
\begin{proof}
We start with devising an algorithm to construct the guarded recursive 
specification $E'$.
The algorithm keeps a set $V$ of recursion equations from $E'$ that are 
already found and a sequence $W$ of equations of the form 
$X_k = \rec{t_k}{E}$ that still have to be transformed.
The algorithm has a finite or countably infinite number of stages.
In each stage, $V$ and $W$ are finite.
Initially, $V$ is empty and $W$ contains only the equation 
$X_0 = \rec{X}{E}$.

In each stage, we remove the first equation from $W$.
Assume that this equation is $X_k = \rec{t_k}{E}$. 
We bring the term $\rec{t_k}{E}$ into head normal form. 
If $t_k$ is not a guarded term, then we use RDP here to turn $t_k$ into 
a guarded term first.
Thus, by Proposition~\ref{prop-HNF-pACPpSIr}, we can always bring 
$\rec{t_k}{E}$ into head normal form.
Assume that the resulting head normal form is
$\PRC{i=1}{n}{\pi_i}
  {(\Altc{j=1}{n_i} a_{ij} \seqc t'_{ij} \altc
    \Altc{k=1}{m_i} b_{ik})}$.
Then, we add the equation 
$X_k = 
 \PRC{i=1}{n}{\pi_i}
  {(\Altc{j=1}{n_i} a_{ij} \seqc X_{k+(\sum_{i'=1}^{i} n_{i'})+j} \altc
    \Altc{k=1}{m_i} b_{ik})}$,
where the $X_{k+(\sum_{i'=1}^{i} n_{i'})+j}$ are fresh variables, to the 
set $V$.
Moreover, for each $i$ and $j$ such that $1 \leq i \leq n$ and 
$1 \leq j \leq n_i$, we add the equation 
$X_{k+(\sum_{i'=1}^{i} n_{i'})+j} = t'_{ij}$ 
to the end of the sequence $W$.
Notice that the terms $t'_{ij}$ are of the form 
$\rec{t_{k+(\sum_{i'=1}^{i} n_{i'})+j}}{E}$.

Because $V$ grows monotonically, there exists a limit. 
That limit is the finite or countably infinite guarded recursive 
specification $E'$.
Every equation that is added to the finite sequence $W$, is also removed 
from it.
Therefore, the right-hand side of each equation from $E'$ only contains
variables that also occur as the left-hand side of an equation from 
$E'$.

\sloppy
Now, we want to use RSP to show that $\rec{X}{E} = \rec{X}{E'}$ is 
derivable from the axioms of \pACPpSIr.
The variables occurring in $E'$ are $X_0, X_1, X_2, \ldots\;$.
For each $k$, the variable $X_k$ has been exactly once in $W$ as the 
left-hand side of an equation.
For each $k$, assume that  this equation is $X_k = \rec{t_k}{E}$.
To use RSP, we have to show for each $k$ that the equation
$X_k = 
 \PRC{i=1}{n}{\pi_i}
  {\linebreak
   (\Altc{j=1}{n_i} a_{ij} \seqc X_{k+(\sum_{i'=1}^{i} n_{i'})+j} \altc
    \Altc{k=1}{m_i} b_{ik})}$,
with, for each $l$, all occurrences of $X_l$ replaced by $\rec{t_l}{E}$,
is derivable from the axioms of \pACPpSIr.
For each $k$, this follows from the construction.
\qed
\end{proof}
Theorem~\ref{thm-expressivity} would not hold if guarded recursive 
specifications were restricted to finite sets of recursion equations.

Let $t$ be a closed \pACP\ term or a closed \pACPpSI\ term, and 
let $X \in \cX$. 
Then $\rec{X}{\set{X = t}} = t$ is derivable from RDP.
This gives rise to the following corollary of 
Theorem~\ref{thm-expressivity}.
\begin{corollary}
\label{corol-expressivity}
For each closed \pACPpSIr\ term $t$, there exists a closed \pACPr\ term 
$t'$ such that $t = t'$ is derivable from the axioms of \pACPpSIr.
\end{corollary}

\subsection{An Example}
\label{subsect-mutex}

In this subsection, we instantiate the generic interleaving strategy on 
which \pACPpSI\ is based with a specific interleaving strategy.
The interleaving strategy concerned corresponds to a scheduling 
algorithm that:
\begin{itemize}
\item
selects randomly, according to a uniform probability distribution, the 
next process that gets turns to perform an action;
\item
gives the selected process a fixed number $k$ of consecutive turns to 
perform an action;
\item
takes care of mutual exclusion of critical subprocesses of the different 
processes being interleaved.
\end{itemize}
Mutual exclusion of certain subprocesses is the condition that they are 
not interleaved and critical subprocesses are subprocesses that possibly 
interfere with each other when this condition is not met.
The adopted mechanism for mutual exclusion is essentially a binary 
semaphore mechanism~\cite{Ben06a,Bri73a,Dij68a}.
Below binary semaphores are simply called \emph{semaphores}.

In this section, it is assumed that a fixed but arbitrary natural number 
$k \in \Natpos$ has been given.
We use $k$ as the number of consecutive turns that each process being
interleaved gets to perform an action.

Moreover, it is assumed that a finite set $R$ of semaphores has been 
given.
\linebreak[2]
We instantiate the set $C$ of control actions as follows: 
\begin{ldispl}
C = \set{\Pop(r) \where r \in R} \union \set{\Vop(r) \where r \in R}\;,
\end{ldispl}%
hereby taking for granted that $C$ satisfies the necessary conditions. 
The $\Pop$ and $\Vop$ actions correspond to the $\mathsf{P}$ and 
$\mathsf{V}$ operations from~\cite{Dij68a}.

We instantiate the set $S$ of control states as follows: 
\begin{ldispl}
S = \Union_{R' \subseteq R} (\mapof{R'}{\seqof{\Natpos\!}})\;.
\end{ldispl}%
The intuition concerning the connection between control states $s \in S$ 
and the semaphore mechanism as introduced in~\cite{Dij68a} is as 
follows:
\pagebreak[2]
\begin{itemize}
\item
$r \notin \dom(s)$ indicates that semaphore $r$ has the value $1$;
\item
$r \in \dom(s)$ indicates that semaphore $r$ has the value $0$;
\item
$r \in \dom(s)$ and $s(r) = \emptyseq$ indicates that no process is 
suspended on \nolinebreak[2] sema\-phore $r$;
\item
if $r \in \dom(s)$ and $s(r) \neq \emptyseq$, then $s(r)$ represents 
a first-in, first-out queue of processes suspended on $r$.
\end{itemize}

As a preparation for the instantiation of the abstract schedulers 
$\sched{n}$ and control state transformers $\updat{n}$, we define some
auxiliary functions.

We define a total function $\funct{\xturns}{\Hist \x \Natpos}{\Nat}$ 
recursively as follows:
\begin{ldispl}
\begin{gceqns}
\xturns(\emptyseq,i) = 0\;,                                          \\
\xturns(h \concat \tup{j,n},i) = 0                & \mif i \neq j\;, \\
\xturns(h \concat \tup{j,n},i) = \xturns(h,i) + 1 & \mif i = j\;.    
\end{gceqns}
\end{ldispl}%
If $\xturns(h,i) = l$ and $l > 0$, then the interleaving history $h$ 
ends with $l$ consecutive turns of the $i$th process being interleaved.
If $\xturns(h,i)= 0$, then the interleaving history $h$ does not end 
with turns of the $i$th process being interleaved.

We define a total function $\funct{\xwaiting}{S}{\setof{\Natpos}}$ as 
follows:
\begin{ldispl}
\begin{geqns}
\xwaiting(s) = \Union_{r \in \dom(s)} \elems(s(r))\;. 
\end{geqns}
\end{ldispl}%
If $\xwaiting(s) = I$, then $i \in I$ iff the $i$th process being 
interleaved is suspended on one or more semaphores in control state $s$.

We define a total function $\funct{\xready{n}}{\Hist \x S}{\set{0,1}}$, 
for each $n \in \Natpos$, as follows:
\begin{ldispl}
\begin{gceqns}
\xready{n}(h,s) = 1
& \mif 
\sum_{i \in \set{1,\ldots,n} \diff \xwaiting(s)}
 \xturns(h,i) \in \set{0,k}\;,
\\ 
\xready{n}(h,s) = 0
& \mif 
\sum_{i \in \set{1,\ldots,n} \diff \xwaiting(s)}
 \xturns(h,i) \notin \set{0,k}\;. 
\end{gceqns}
\end{ldispl}%
If $\xready{n}(h,s) = b$, then $b = 1$ iff the interleaving history $h$ 
ends with a number of consecutive turns of some process that equals $k$ 
if that process is not suspended in control state $s$.

We define a partial function 
$\pfunct{\xsched{n}}{\Hist \x S}{(\set{1,\ldots,n} \to \Prob)}$, 
for each $n \in \Natpos$, as follows:
\begin{ldispl}
\begin{geqns}
\xsched{n}(h,s)(i) = 1 / (n - \card(\xwaiting(s))) \\ \qquad
\mif \xready{n}(h,s) = 1 \Land i \notin \xwaiting(s) \Land
     \xwaiting(s) \neq \set{1,\ldots,n}\;, 
\\
\xsched{n}(h,s)(i) = 0 \\ \qquad
\mif \xready{n}(h,s) = 1 \Land i \in \xwaiting(s) \Land
     \xwaiting(s) \neq \set{1,\ldots,n}\;, 
\\   
\xsched{n}(h,s)(i) = 1 \\ \qquad
\mif \xready{n}(h,s) = 0 \Land \xturns(h,i) \neq 0 \Land
     \xwaiting(s) \neq \set{1,\ldots,n}\;, 
\\   
\xsched{n}(h,s)(i) = 0 \\ \qquad
\mif \xready{n}(h,s) = 0 \Land \xturns(h,i) = 0 \Land
     \xwaiting(s) \neq \set{1,\ldots,n}\;. 
\end{geqns}
\end{ldispl}%
The function $\xsched{n}$ represents a scheduler that work as follows: 
when a process has been given $k$ consecutive turns to perform an action 
or has been suspended, the next process that is given turns is randomly 
selected, according to a uniform probability distribution, from the 
processes being interleaved that are not suspended.
Notice that $\xsched{n}(h,s)(i)$ is undefined if 
$\xwaiting(s) = \set{1,\ldots,n}$.
In that case, none of the processes being interleaved can be given a 
turn and the whole becomes inactive.

We define a total function 
$\funct{\xremove{n}}{S \x \set{1,\ldots,n}}{S}$ 
recursively as follows:%
\footnote
{The special function notation used in this paper is explained in an
 appendix.}
\begin{ldispl}
\begin{geqns}
\xremove{n}(\emptymap,i) = \emptymap\;, \\
\xremove{n}(s \owr \maplet{r}{q},i) = 
\xremove{n}(s,i) \owr \maplet{r}{\xremovep{n}(q,i)}\;,
\end{geqns}
\end{ldispl}%
where the total function 
$\funct{\xremovep{n}}
  {\seqof{\Natpos} \x \set{1,\ldots,n}}{\seqof{\Natpos}}$
is recursively defined as follows:
\begin{ldispl}
\begin{gceqns}
\xremovep{n}(\emptyseq,i) = \emptyseq\;, \\
\xremovep{n}(j \concat q,i) = j \concat \xremovep{n}(q,i)
                                                    & \mif j < i\;, \\
\xremovep{n}(j \concat q,i) = \xremovep{n}(q,i)     & \mif j = i\;, \\
\xremovep{n}(j \concat q,i) = (j - 1) \concat \xremovep{n}(q) 
                                                    & \mif j > i\;.
\end{gceqns}
\end{ldispl}%
If $\xremove{n}(s,i) = s'$, then $s'$ is $s$ adapted to the successful
termination of the $i$th process of the processes being interleaved.

For each $n \in \Natpos$, we instantiate the abstract scheduler 
$\sched{n}$ and control state transformer $\updat{n}$ as follows:
\begin{ldispl}
\begin{gceqns}
\sched{n}(h,s) = \xsched{n}(h,s)\;,
\eqnsep
\updat{n}(\emptyseq,s,i,a,0) = \emptymap     & \mif a \notin C\;,     \\
\updat{n}(h \concat \tup{j,n},s,i,a,0) = s   & \mif a \notin C\;,     \\
\updat{n}(\emptyseq,s,i,\Pop(r),0) = \maplet{r}{\emptyseq}\;,         \\
\updat{n}(h \concat \tup{j,n},s,i,\Pop(r),0) = 
s \owr \maplet{r}{\emptyseq}               & \mif r \notin \dom(s)\;, \\
\updat{n}(h \concat \tup{j,n},s,i,\Pop(r),0) = 
s \owr \maplet{r}{s(r) \concat i}          & \mif r \in \dom(s)\;,    \\
\updat{n}(\emptyseq,s,i,\Vop(r),0) = \emptymap\;,                     \\
\updat{n}(h \concat \tup{j,n},s,i,\Vop(r),0) = s             
                                           & \mif r \notin \dom(s)\;, \\
\updat{n}(h \concat \tup{j,n},s,i,\Vop(r),0) = s \dsub \set{r}            
                    & \mif r \in \dom(s) \Land s(r) = \emptyseq\;,    \\
\updat{n}(h \concat \tup{j,n},s,i,\Vop(r),0) = 
s \owr \maplet{r}{\tl(s(r))}             
                    & \mif r \in \dom(s) \Land s(r) \neq \emptyseq\;, \\
\updat{n}(h,s,i,a,1) = \xremove{n}(s,i)\;.
\end{gceqns}
\end{ldispl}%
The following clarifies the connection between the instantiated control 
state transformers $\updat{n}$ and the semaphore mechanism as introduced 
in~\cite{Dij68a}:
\begin{itemize}
\item
$s = \emptymap$ indicates that all semaphores have value $1$;
\item
if $r \notin \dom(s)$, then the transition from $s$ to 
$s \owr \maplet{r}{\emptyseq}$ indicates that the value of semaphore $r$ 
changes from $1$ to $0$;
\item
if $r \in \dom(s)$, then the transition from $s$ to 
$s \owr \maplet{r}{s(r) \concat i}$ indicates that \linebreak[2] the 
$i$th process being interleaved is added to the queue of processes 
suspended on semaphore $r$;
\item
if $r \notin \dom(s)$, then the transition from $s$ to $s$ indicates 
that the value of semaphore $r$ remains $1$;
\item
if $r \in \dom(s)$ and $s(r) = \emptyseq$, then the transition from $s$ 
to $s \dsub \set{r}$ indicates that the value of semaphore $r$ changes 
from $0$ to $1$;
\item
if $r \in \dom(s)$ and $s(r) \neq \emptyseq$, then the transition from 
$s$ to $s \owr \maplet{r}{\tl(s(r))}$ \linebreak[2] indicates that the 
first process in the queue of processes suspended on semaphore $r$ is 
removed from that queue.
\end{itemize}

The example given above is only meant to show that the generic 
probabilistic interleaving strategy assumed in \pACPpSI\ can be 
instantiated with non-trivial specific probabilistic interleaving 
strategies.
In practice, more advanced probabilistic interleaving strategies, such 
as strategies based on lottery scheduling~\cite{WW94a}, are more 
important.

\section{Concluding Remarks}
\label{sect-concl}

We have presented a probabilistic version of \ACP~\cite{BW90,BK84b} that 
rests on the principle that probabilistic choices are always resolved 
before choices involved in alternative composition and parallel 
composition are resolved.
By taking functions whose range is the carrier of a signed cancellation 
meadow~\cite{BBP13a,BT07a} instead of a field as probability measures, 
we could include probabilistic choice operators for the probabilities 
$0$ and $1$ without any problem and give a simple operational semantics.

We have also extended this probabilistic version of \ACP\ with a form of 
interleaving in which parallel processes are interleaved according to 
what is known as a process-scheduling policy in the field of operating 
systems. 
This is the form of interleaving that underlies multi-threading as found 
in contemporary programming languages.
To our knowledge, the work presented in~\cite{BM17b} and this paper is 
the only work on this form of interleaving in the setting of a general 
algebraic theory of processes like ACP, CCS and CSP.

The main probabilistic versions of \ACP\ introduced earlier are 
prACP~\cite{BBS95a}, \pACPpa~\cite{And99a}, and \pACPt~\cite{AG09a}.
Like \pACP, those probabilistic versions of \ACP\ are based on the 
generative model of probabilistic processes.
In prACP, the alternative composition operator and the parallel 
composition operator are replaced by probabilistic choice operators and 
probabilistic parallel composition operators.
In \pACPpa, no operators are replaced, but probabilistic choice 
operators are added.
The parallel composition operator of \pACPpa\ is somewhat tricky because 
probabilistic choices are not resolved before choices involved in 
parallel composition are resolved.
\pACPt\ is, apart from abstraction, \pACPpa\ with another parallel 
composition operator where probabilistic choices are resolved before 
choices involved in parallel composition are resolved.
\pACP\ is a minor variant of \pACPt\ without abstraction operators.
The differences and their consequences are described in the first and 
last but one paragraph of Section~\ref{subsect-remarks}.

In this paper, we consider strategic interleaving where process creation 
is taken into account.
The approach to process creation followed originates from the one first 
followed in~\cite{Ber90a} to extend \ACP\ with process creation and 
later followed in~\cite{BB93a,BM02a,BMU98a} to extend different timed 
versions of \ACP\ with process creation.
The only other approach that we know of is the approach, based 
on~\cite{AB88a}, that has for instance been followed 
in~\cite{BV92a,GR97a}.
However, with that approach, it is most unlikely that data about the
creation of processes can be made available for the decision making 
concerning the strategic interleaving of processes.

\appendix

\section*{Appendix: Sequence Notation and Function Notation}
\label{appendix-notations}

We use the following sequence notation:
\begin{itemize}
\item
$\emptyseq$ for the empty sequence;
\item 
$d$ for the sequence having $d$ as sole element;
\item
$u \concat v$ for the concatenation of sequences $u$ and $v$;
\item
$\hd(u)$ for the first element of non-empty sequence $u$;
\item
$\tl(u)$ for the subsequence of non-empty sequence $u$ whose first 
element is the second element of $u$ and whose last element is the last 
element of $u$;
\item
$\elems(u)$ is the set of all elements of sequence $u$. 
\end{itemize}
\pagebreak[2]
We use the following special function notation:
\begin{itemize}
\item
$\emptymap$ for the empty function; 
\item
$\maplet{d}{e}$ for the function $f$ with $\dom(f) = \set{d}$ such that
$f(d) = e$; 
\item
$f \owr g$ for the function $h$ with $\dom(h) = \dom(f) \union \dom(g)$
such that for all $d \in \dom(h)$, $h(d) = f(d)$ if
$d \notin \dom(g)$ and $h(d) = g(d)$ otherwise;
\item
$f \dsub S$ for the function $g$ with $\dom(g) = \dom(f) \diff S$
such that for all $d \in \dom(g)$, $g(d) = f(d)$.
\end{itemize}

\bibliographystyle{splncs03}
\bibliography{PA}

\begin{thebibliography}{10}
\providecommand{\url}[1]{\texttt{#1}}
\providecommand{\urlprefix}{URL }

\bibitem{AB88a}
America, P., de~Bakker, J.W.: Designing equivalent semantic models for process
  creation. Theoretical Computer Science  60(2),  109--176 (1988)

\bibitem{And99a}
Andova, S.: Process algebra with probabilistic choice. In: Katoen, J.P. (ed.)
  ARTS'99. Lecture Notes in Computer Science, vol. 1601, pp. 111--129.
  Springer-Verlag (1999)

\bibitem{And02a}
Andova, S.: Probabilistic Process Algebra. Ph.D. thesis, Department of
  Mathematics and Computer Science, Eindhoven University of Technology,
  Eindhoven (2002)

\bibitem{AG09a}
Andova, S., Georgievska, S.: On compositionality, efficiency, and applicability
  of abstraction in probabilistic systems. In: Nielsen, M., et~al. (eds.)
  SOFSEM 2009. Lecture Notes in Computer Science, vol. 5404, pp. 67--78.
  Springer-Verlag (2009)

\bibitem{BB93a}
Baeten, J.C.M., Bergstra, J.A.: Real space process algebra. Formal Aspects of
  Computing  5(6),  481--529 (1993)

\bibitem{BBS95a}
Baeten, J.C.M., Bergstra, J.A., Smolka, S.A.: Axiomatizing probabilistic
  processes: {ACP} with generative probabilities. Information and Computation
  121(2),  234--255 (1995)

\bibitem{BM02a}
Baeten, J.C.M., Middelburg, C.A.: Process Algebra with Timing. Monographs in
  Theoretical Computer Science, An EATCS Series, Springer-Verlag, Berlin (2002)

\bibitem{BV92a}
Baeten, J.C.M., Vaandrager, F.W.: An algebra of process creation. Acta
  Informatica  29(4),  303--334 (1992)

\bibitem{BW90}
Baeten, J.C.M., Weijland, W.P.: Process Algebra, Cambridge Tracts in
  Theoretical Computer Science, vol.~18. Cambridge University Press, Cambridge
  (1990)

\bibitem{Ben06a}
Ben-Ari, M.: Principles of Concurrent and Distributed Programming. Pearson,
  Harlow, second edn. (2006)

\bibitem{Ber90a}
Bergstra, J.A.: A process creation mechanism in process algebra. In: Baeten,
  J.C.M. (ed.) Applications of Process Algebra, Cambridge Tracts in Theoretical
  Computer Science, vol.~17, pp. 81--88. Cambridge University Press, Cambridge
  (1990)

\bibitem{BBP13a}
Bergstra, J.A., Bethke, I., Ponse, A.: Cancellation meadows: A generic basis
  theorem and some applications. Computer Journal  56(1),  3--14 (2013)

\bibitem{BK84a}
Bergstra, J.A., Klop, J.W.: The algebra of recursively defined processes and
  the algebra of regular processes. In: Paredaens, J. (ed.) Proceedings 11th
  ICALP. Lecture Notes in Computer Science, vol. 172, pp. 82--95.
  Springer-Verlag (1984)

\bibitem{BK84b}
Bergstra, J.A., Klop, J.W.: Process algebra for synchronous communication.
  Information and Control  60(1--3),  109--137 (1984)

\bibitem{BM04c}
Bergstra, J.A., Middelburg, C.A.: Thread algebra for strategic interleaving.
  Formal Aspects of Computing  19(4),  445--474 (2007)

\bibitem{BM17b}
Bergstra, J.A., Middelburg, C.A.: Process algebra with strategic interleaving.
  Theory of Computing Systems  63(3),  488--505 (2019)

\bibitem{BMU98a}
Bergstra, J.A., Middelburg, C.A., Usenko, Y.S.: Discrete time process algebra
  and the semantics of {SDL}. In: Bergstra, J.A., Ponse, A., Smolka, S.A.
  (eds.) Handbook of Process Algebra, pp. 1209--1268. Elsevier, Amsterdam
  (2001)

\bibitem{BP13a}
Bergstra, J.A., Ponse, A.: Probability functions in the context of signed
  involutive meadows. In: James, P., Roggenbach, M. (eds.) WADT 2016. Lecture
  Notes in Computer Science, vol. 10644, pp. 73--87. Springer-Verlag (2017)

\bibitem{BT07a}
Bergstra, J.A., Tucker, J.V.: The rational numbers as an abstract data type.
  Journal of the ACM  54(2),  Article 7 (2007)

\bibitem{Bri73a}
{Brinch Hansen}, P.: Operating System Principles. Prentice-Hall, Englewood
  Cliffs, NJ (1973)

\bibitem{Dij68a}
Dijkstra, E.W.: Cooperating sequential processes. In: Genuys, F. (ed.)
  Programming Languages. pp. 43--112. Academic Press (1968)

\bibitem{GR97a}
Gehrke, T., Rensink, A.: Process creation and full sequential composition in a
  name-passing calculus. Electronic Notes in Theoretical Computer Science  7,
  141--160 (1997)

\bibitem{Geo11a}
Georgievska, S.: Probability and Hiding in Concurrent Processes. Ph.D. thesis,
  Department of Mathematics and Computer Science, Eindhoven University of
  Technology, Eindhoven (2011)

\bibitem{GSS95a}
van Glabbeek, R.J., Smolka, S.A., Steffen, B.: Reactive, generative and
  stratified models of probabilistic processes. Information and Computation
  121(1),  59--80 (1995)

\bibitem{GV93}
van Glabbeek, R.J., Vaandrager, F.W.: Modular specification of process
  algebras. Theoretical Computer Science  113(2),  293--348 (1993)

\bibitem{GJSB00a}
Gosling, J., Joy, B., Steele, G., Bracha, G.: The {Java} Language
  Specification. Addison-Wesley, Reading, MA, second edn. (2000)

\bibitem{HWG03a}
Hejlsberg, A., Wiltamuth, S., Golde, P.: {C\#} Language Specification.
  Addison-Wesley, Reading, MA (2003)

\bibitem{Hoa85}
Hoare, C.A.R.: Communicating Sequential Processes. Prentice-Hall, Englewood
  Cliffs (1985)

\bibitem{LT09a}
Lanotte, R., Tini, S.: Probabilistic bisimulation as a congruence. ACM
  Transactions on Computational Logic  10(2),  Article 9 (2009)

\bibitem{Mil89}
Milner, R.: Communication and Concurrency. Prentice-Hall, Englewood Cliffs
  (1989)

\bibitem{SS00a}
Sabelfeld, A., Sands, D.: Probabilistic noninterference for multi-threaded
  programs. In: Computer Security Foundations Workshop 2000. pp. 200--214. IEEE
  Computer Society Press (2000)

\bibitem{SGG18a}
Silberschatz, A., Galvin, P.B., Gagne, G.: Operating System Concepts. John
  Wiley and Sons, Hoboken, NJ, tenth edn. (2018)

\bibitem{TB15a}
Tanenbaum, A.S., Bos, H.: Modern Operating Systems. Pearson, Harlow, fourth
  edn. (2015)

\bibitem{WW94a}
Waldspurger, C.A., Weihl, W.E.: Lottery scheduling: Flexible proportional-share
  resource management. In: OSDI '94. p. Article 1. USENIX Association, Berkeley
  (1994)

\end{thebibliography}

\end{document}